\documentclass[preprint]{elsarticle}

\usepackage{subfig}
\usepackage{mathtools}
\usepackage{amsmath}
\usepackage{amssymb}
\usepackage{amsthm}
\usepackage{array}
\usepackage{enumerate}
\usepackage{float}
\usepackage{mathtools}
\usepackage{pifont}
\usepackage{hyperref}
\usepackage{url}
\usepackage{tikz}
\usepackage{pgfplots}
\usepackage{colortbl}
\usepackage{color}
\usetikzlibrary{external}
\usetikzlibrary{shapes}

\newcommand{\reusefigure}{} 

\newcommand{\includepgfplots}[1]{%
  \ifx\reusefigure\undefined
      \tikzsetnextfilename{#1}%
      \input{src/img/#1.tex}%
  \else
      \includegraphics{img/#1}
  \fi
}

\definecolor{green_new}{RGB}{0,128,0}
\usepgfplotslibrary{fillbetween}
\newcommand{\greencircle}{\raisebox{0.5pt}{\tikz{\node[draw,scale=0.4,circle,fill=black!20!green_new](){};}}}
\newcommand{\bluediamond}{\raisebox{0pt}{\tikz{\node[draw,scale=0.4,diamond,fill=black!20!blue](){};}}}

\hypersetup{colorlinks,
  linkcolor={red!50!black},
  citecolor={blue!50!black},
  urlcolor={blue!80!black}}

\newcommand{\xmark}{\ding{55}}%

\newtheorem{thm}{Theorem}
\newtheorem{lemma}[thm]{Lemma}

\newtheorem{defin}{Definition}
\newtheorem{prop}{Proposition}
\newtheorem{situation}{Situation}
\newtheoremstyle{algorithm}
  {\topsep}
  {\topsep}
  {}
  {}
  {\scshape}
  {.}
  { }
  {\thmname{#1}\thmnumber{ #2}\thmnote{ (#3)}}
\theoremstyle{algorithm}
\newtheorem{alg}{Algorithm}
\journal{Neurocomputing}

\begin{document}

\begin{frontmatter}
\title{``Parallel Training Considered Harmful?'':\\ Comparing series-parallel and parallel feedforward network training}
\author{Ant\^{o}nio H. Ribeiro\corref{cor1}}
 \address{Graduate  Program in  Electrical  Engineering at
   Universidade Federal de Minas Gerais (UFMG) -
  Av. Ant\^{o}nio Carlos 6627, 31270-901, Belo Horizonte, MG, Brazil}
 \cortext[cor1]{This work has been supported by the Brazilian agencies CAPES, CNPq and FAPEMIG.}
 \ead{antonio-ribeiro@ufmg.br}
 \ead[url]{antonior92.github.io}
 \author{Luis A. Aguirre}
 \address{Department of Electronic Engineering at
   Universidade Federal de Minas Gerais (UFMG) -
  Av. Ant\^{o}nio Carlos 6627, 31270-901, Belo Horizonte, MG, Brazil}
\ead{aguirre@ufmg.br}

\begin{abstract}
  Neural network models for dynamic systems can be trained either in parallel or in series-parallel configurations. Influenced by early arguments, several papers justify the choice of series-parallel rather than parallel configuration claiming it has a lower computational cost, better stability properties during training and provides more accurate results. Other published results, on the other hand, defend parallel training as being more robust and capable of yielding more accurate long-term predictions. The main contribution of this paper is to present a study comparing both methods under the same unified framework with special attention to three aspects: i)~robustness of the estimation in the presence of noise; ii)~computational cost; and, iii)~convergence. A unifying mathematical framework and simulation studies show situations where each training method provides superior validation results and suggest that parallel training is generally better in more realistic scenarios. An example using measured data seems to reinforce such a claim. Complexity analysis and numerical examples show that both methods have similar computational cost although series-parallel training is more amenable to parallelization. Some informal discussion about stability and convergence properties is presented and explored in the examples.
\end{abstract}

\end{frontmatter}

\section{Introduction}

Neural networks are widely used and studied for modeling nonlinear dynamic
systems~\cite{narendra_identification_1990, chen_non-linear_1990}.
In the seminal paper by Narenda and Parthasarathy~\cite{narendra_identification_1990}
series-parallel and parallel configurations were introduced as possible architectures for training neural networks with data from dynamic systems.
In both cases the neural network parameters are estimated by minimizing the error
between predicted and measured values.
When training the neural network in series-parallel configuration
{\it measured}\, values from past instants are used to make \textit{one-step-ahead} predictions.
On the other hand, when training in parallel configuration previously 
{\it predicted}\, outputs are fed back into the network to compute what is known as the \textit{free-run simulation} of the network.

In a very influential paper~\cite{narendra_identification_1990}, training in series-parallel configuration is said
to be preferable to parallel configuration for three main reasons:
(i)~all signals generated in the identification procedure for series-parallel configuration
are bounded, while this is not guaranteed for the parallel
configuration (this argument is also mentioned in~\cite{zhang_modeling_2006, singh_identification_2013});
(ii)~for the parallel configuration a modified version of backpropagation is needed, resulting
in a greater computational cost for training, while the standard backpropagation can be used in the context of the
series-parallel configuration (this has also been mentioned in~\cite{saad_adaptive_1994, beale_neural_2017,
  saggar_system_2007}); and,
(iii)~assuming that the error tends asymptotically to small values,
the simulated output is not so different from the actual one and,
therefore, the results obtained from two configurations would not be significantly
different (this reasoning was repeated in~\cite{warwick_introduction_1996, kaminnski_genetic_1996, rahman_neural_2000}).
An additional reason that also appears sometimes in the literature is
that: (iv)~the series-parallel training provides better results because of
more accurate inputs to the neural network during training (this argument has also been used
in~\cite{beale_neural_2017,
  saggar_system_2007, petrovic_kalman_2013, tijani_nonlinear_2014, khan_forecasting_2015,
  diaconescu_use_2008}).

It is clear that the seminal paper \cite{narendra_identification_1990} still has great influence on the literature. However, a number of subtle aspects related to the alleged reasons for preferring series-parallel training were not discussed in that paper and, from the  literature, seem to remain unnoticed in general. This may have important practical consequences especially in situations where parallel training would be the best option. Hence, one of the aims of this work is to establish the validity of such arguments.

Other papers followed as somewhat complimentary track highlighting the strengths of parallel training
in system identification. For instance, according to~\cite{su_long-term_1992,
  su_neural_1993} neural networks trained in parallel yield more accurate
long-term predictions than the ones trained in series-parallel; it has been shown that
for diverse types of models, including neural networks, parallel training can be more
robust than series-parallel training in the presence of some types of noise \cite{aguirre_prediction_2010};
and, in a series of papers Piroddi and co-workers have argued that parallel-training
(i.e. free-run simulation error minimization) is a promising technique for structure selection of polynomial models~\cite{piroddi_identification_2003, farina_convergence_2008,
  farina_iterative_2010, farina_simulation_2011,
  farina_identification_2012}. Some papers also compare the two training methods in practical
settings: in~\cite{patan_nonlinear_2012}
neural network parallel models present better validation results than
series-parallel models for modeling a boiler unit; and, in~\cite{zhang_new_2014}
parallel training provided the best results when estimating the parameters of a battery.

The main contribution of this paper is to provide a detailed comparison of parallel and series-parallel training. Three aspects receive special attention: (i) computational cost; (ii) robustness of the estimation in the presence of noise; and, (iii) convergence.
The rest of the paper is organized as follows:
Section~\ref{sec:background} reviews backpropagation and nonlinear least-squares algorithms.
Section~\ref{sec:nntraining} provides the nonlinear least-squares framework used along this paper.
A complexity analysis is presented in Section~\ref{sec:complexity-analysis}.
Section~\ref{sec:unifying-framework} give asymptotic properties
of the methods in the presence noise. In Section~\ref{sec:examples}, examples
are used to investigate the effect of the noise and the training time.
Section~\ref{sec:discussion} discusses signal unboundedness, the convergence to ``poor''
local solutions and computational considerations.
Final comments and future work are presented
in Section~\ref{sec:concl-future-work}.

\section{Background}
\label{sec:background}

This section provides a quick review of algorithms and ideas
used to build our formulation:
Section~\ref{sec:parall-text-seri}
explains series-parallel and parallel training;
Section~\ref{sec:non-linear-least-squares} explains the Levenberg-Marquardt
algorithm~\cite{marquardt_algorithm_1963}, presenting a formulation similar to
the one in~\cite{fletcher_modified_1971};
Section~\ref{sec:modif-backpr} explains a modified backpropagation algorithm,
proposed in~\cite{hagan_training_1994}.

\subsection{Parallel \textit{vs} series-parallel training}
\label{sec:parall-text-seri}
Figure~\ref{fig:parallel-vs-seriesparallel} illustrates the difference
between parallel and series-parallel training. \textit{Parallel configuration} feeds the
output back to the input of the network and, hence, uses its own previous
values to predict the next output of the system being modeled.
\textit{Series-parallel configuration}, on the other hand,
uses the true measured output rather than feeding back the estimated
one. We formalize these concepts along this section.

\begin{figure}[htpb]
  \centering
  \subfloat[Parallel]{
    \includegraphics[width=0.43\textwidth]{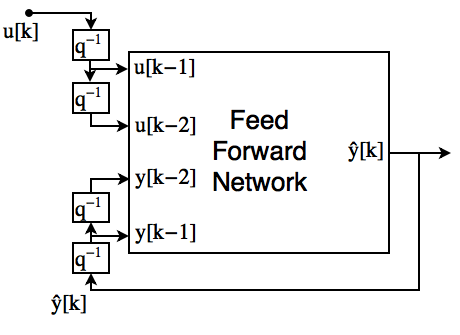}}
  \subfloat[Series-parallel]{
    \includegraphics[width=0.45\textwidth]{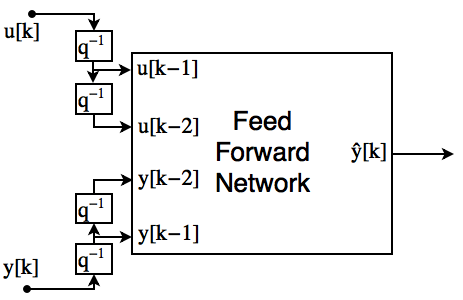}}
  \caption{
    Parallel and series-parallel neural network architectures for
    modeling the dynamic system
    $\mathbf{y}[k] = \mathbf{F}(\mathbf{y}[k-1], \mathbf{y}[k-2],
    \mathbf{u}[k-1], \mathbf{u}[k-2])$. The delay operator $q^{-1}$ is such that
    $\mathbf{y}[k-1]=q^{-1}\mathbf{y}[k]$.}
  \label{fig:parallel-vs-seriesparallel}
\end{figure}

Consider the following dataset:
$
  \mathcal{Z} = \{(\mathbf{u}[k], \mathbf{y}[k]), k = 1, 2, \hdots, N\} ,
$
containing  a sequence of sampled inputs-output pairs.
Here ${\mathbf{u}[k]\in \mathbb{R}^{N_u}}$ and ${\mathbf{y}[k]\in \mathbb{R}^{N_y}}$
are vectors containing all the inputs and outputs of interest at instant $k$.
The output $\mathbf{y}[k]$ is correlated with its own past
values $\mathbf{y}[k-1]$, $\mathbf{y}[k-2]$, $\mathbf{y}[k-3]$, $\cdots$, and with
past input values $\mathbf{u}[k-\tau_d]$,
$\mathbf{u}[k-\tau_d-1]$, $\mathbf{u}[k-\tau_d-2]$, $\cdots$. The integer
$\tau_d \ge 0$ is the input-output delay (when changes in the input take some time
to affect the output).


The aim is to find a difference equation model:
\begin{equation}
  \label{eq:diff_eq}
    \mathbf{y}[k] = \mathbf{F}(\mathbf{y}[k-1], \hdots, \mathbf{y}[k-n_y],
                       \mathbf{u}[k-\tau_d], \hdots,
                       \mathbf{u}[k-n_u]; \mathbf{\Theta}),
\end{equation}                    
that best describes the data in $\mathcal{Z}$.
The model is described by the nonlinear function $\mathbf{F}$ parameterized
by ${\mathbf{\Theta}\in \mathbb{R}^{N_\Theta}}$;
the maximum input and output lags $n_y$ and $n_u$;
and, the input-output delay $\tau_d$. It is assumed that a finite number of past terms can
be used to describe the output.

Neural networks are an appropriate choice for representing the function $\mathbf{F}$
because they are universal approximators. A neural network
with as few as one hidden layer can approximate any
measurable function $\mathbf{F}^*$ with arbitrary accuracy
within a given compact set~\cite{hornik_multilayer_1989}.

From now on, the simplified notation:
${\footnotesize \underline{\mathbf{y}}_{[k]} = [\mathbf{y}[k-1]^T, \hdots,
  \mathbf{y}[k-n_y]^T]^T}$, \\
${\footnotesize \underline{\mathbf{u}}_{[k]} = [\mathbf{u}[k-\tau_d]^T, \hdots,
  \mathbf{u}[k-n_u]^T]^T}$
will be used, hence,
Equation~(\ref{eq:diff_eq}) can be rewritten as:
$\mathbf{y}[k] = \mathbf{F}(\underline{\mathbf{y}}_{[k]}, \underline{\mathbf{u}}_{[k]} ; \mathbf{\Theta})$.

\begin{defin}[One-step-ahead prediction]
  For a given function $\mathbf{F}$, parameter vector $\mathbf{\Theta}$ 
  and dataset $\mathcal{Z}$, the one-step-ahead prediction is defined as:
  \begin{equation}
    \label{eq:one-step-ahead-prediction}
    \hat{\mathbf{y}}_1[k]  = \mathbf{F}(\underline{\mathbf{y}}_{[k]},
    \underline{\mathbf{u}}_{[k]}; \mathbf{\Theta}),
  \end{equation}
\end{defin}
\begin{defin}[Free-run simulation]
  For a given function $\mathbf{F}$, parameter vector $\mathbf{\Theta}$,
  dataset $\mathcal{Z}$ and a set of initial conditions $\{\mathbf{y}_0[k]\}_{k=1}^{n_y}$,
  the free-run simulation is defined using the recursive formula:
  \begin{equation}
    \label{eq:free_run_simulation_recursion}
    \hat{\mathbf{y}}_s[k] =
    \begin{dcases}
      \mathbf{y}_0[k],~1 \le  k\le n_y; \\
        \mathbf{F}(\hat{\mathbf{y}}_s[k-1], \hdots,  \hat{\mathbf{y}}_s[k-n_y],
        \underline{\mathbf{u}}_{[k]}; \mathbf{\Theta}),~
        k> n_y,
    \end{dcases}
  \end{equation}
  \noindent
  The vector of initial conditions is defined as
  ${\underline{\mathbf{y}}_0=[\mathbf{y}_0[1]^T}$, ${\hdots, \mathbf{y}_0[n_y]^T]^T}$.
\end{defin}

Both training procedures minimize some norm of the error and may be regarded
as \textit{prediction error methods}.\footnote{In the context of predictor error methods the
  nomenclature NARX (\textit{nonlinear autoregressive model with exogenous input})
  and NOE (\textit{nonlinear output error model}) is often used to refer to the models
  obtained using, respectively, series-parallel and parallel
  training.}
That is, let ${\mathbf{e}=[\mathbf{e}[1]^T, \hdots, \mathbf{e}[N]^T]^T}$,
in this paper we consider the parameters are estimated minimizing 
the sum of square errors $\tfrac{1}{2}\|\mathbf{\mathbf{e}}\|^2$.\footnote{
  This loss function
  is optimal in the maximum likelihood sense
  if the residuals are considered to be Gaussian white
  noise~\cite{nelles_nonlinear_2013}.}
The one-step-ahead error,
${\mathbf{e}_1[k]={\hat{\mathbf{y}}_1[k] -\mathbf{y}[k]}}$, is used for series-parallel
training and the free-run simulation
error, ${\mathbf{e}_s[k]=\hat{\mathbf{y}}_s[k] - \mathbf{y}[k]}$, for parallel
training.

\subsection{Nonlinear least-squares}
\label{sec:non-linear-least-squares}

Let $\mathbf{\Theta}\in\mathbb{R}^{N_\Theta}$ be a vector of
parameters and $\mathbf{e}(\mathbf{\Theta}) \in \mathbb{R}^{N_e}$ an error vector.
In order to estimate the parameter vector $\mathbf{\Theta}$ the
sum of square errors
$ V(\mathbf{\Theta}) = \tfrac{1}{2}\|\mathbf{e}(\mathbf{\Theta})\|^2$ is minimized.
Its gradient vector and Hessian matrix may be computed as:~\cite[p. 246]{nocedal_numerical_2006}
\begin{eqnarray}
  \label{eq:gradient}
  \tfrac{\partial V}{\partial \mathbf{\Theta}}
  &=&  \left[\tfrac{\partial \mathbf{e}}{\partial \mathbf{\Theta}}\right]^T
      \mathbf{e}(\mathbf{\Theta}),\\
  \label{eq:hessian}
  \tfrac{\partial^2 V}{\partial \mathbf{\Theta}^2} &=&
  \left[\tfrac{\partial \mathbf{e}}{\partial \mathbf{\Theta}}\right]^T
  \left[\tfrac{\partial \mathbf{e}}{\partial \mathbf{\Theta}}\right]
  + \sum_{i=1}^{N_e} e_i  \tfrac{\partial^2e_i}{\partial\mathbf{\Theta}^2},
\end{eqnarray}
where
${\tfrac{\partial \mathbf{e}}{\partial \mathbf{\Theta}} \in\mathbb{R}^{N_e \times N_\Theta}}$
is  the Jacobian matrix associated with $\mathbf{e}(\mathbf{\Theta})$.
Non-linear least-squares algorithms usually update the solution iteratively
  ($\mathbf{\Theta}^{n+1} = \mathbf{\Theta}^n + \Delta \mathbf{\Theta}^n$)
and exploit the special structure of the gradient and Hessian
of $V(\mathbf{\Theta})$, in order to compute the parameter update
$\Delta \mathbf{\Theta}^n$.

The Levenberg-Marquardt algorithm considers a parameter
update~\cite{marquardt_algorithm_1963}:
\begin{equation}
  \label{eq:levenberg_marquardt}
  \Delta \mathbf{\Theta}^n =
  -{\left[\left.\tfrac{\partial \mathbf{e}}{\partial \mathbf{\Theta}}\right|^T_n
      \left.\tfrac{\partial \mathbf{e}}{\partial \mathbf{\Theta}}\right|_n
      + \lambda^n D^n\right]}^{-1} 
  \left[\left.\tfrac{\partial \mathbf{e}}{\partial \mathbf{\Theta}}\right|_n
  \right]^T \mathbf{e}_n,
\end{equation}
for which $\lambda^n$ is a non-negative scalar and $D^n$ is a non negative
diagonal matrix. Furthermore, $\mathbf{e}_n$ and $\left.\tfrac{\partial \mathbf{e}}{\partial \mathbf{\Theta}}\right|_n$
are the error and the corresponding Jacobian matrix evaluated at $\mathbf{\Theta}^n$.

There are different ways of updating $\lambda^n$ and $D^n$.
The update strategy presented here
is similar to~\cite{fletcher_modified_1971}.
The elements of the diagonal matrix $D^n$ are chosen equal to the
elements in the diagonal of
$\left.\tfrac{\partial \mathbf{e}}{\partial \mathbf{\Theta}}\right|^T_n
      \left.\tfrac{\partial \mathbf{e}}{\partial \mathbf{\Theta}}\right|_n$.
And $\lambda^n$ is increased or decreased according to the agreement between
the local model
$\left(\phi_n(\Delta \mathbf{\Theta}^n) = \tfrac{1}{2}\left\|\mathbf{e}_n +
  \left[\left. \tfrac{\partial \mathbf{e}}{\partial
        \mathbf{\Theta}}\right|_n\right]\Delta \mathbf{\Theta}^n\right\|^2\right)$
and the real objective function $V(\mathbf{\Theta}_n)$. The degree of agreement is measured
using the following ratio:
\begin{equation}
  \label{eq:reduction_ratio}
  \rho_n = \frac{V(\mathbf{\Theta}^n) - V(\mathbf{\Theta}^n + \Delta \mathbf{\Theta}^n)}{\phi_n(\mathbf{0}) - \phi_n(\Delta \mathbf{\Theta}^n)}.
\end{equation}
One iteration of the algorithm is summarized next:
\begin{alg}[Levenberg-Marquardt Iteration]
  \label{alg:levenberg-marquardt}
For a given $\mathbf{\Theta}^n$ and $\lambda^n$:
\begin{enumerate}
\item Compute $\mathbf{e}(\mathbf{\Theta}^n)$ and $\tfrac{\partial \mathbf{e}}{\partial \mathbf{\Theta}}(\mathbf{\Theta}^n)$, if not already computed.
\item $\text{diag}({D}^n)\leftarrow \text{diag}\left(\left.\tfrac{\partial \mathbf{e}}{\partial \mathbf{\Theta}}\right|^T_n
      \left.\tfrac{\partial \mathbf{e}}{\partial \mathbf{\Theta}}\right|_n\right)$.
\item Solve~(\ref{eq:levenberg_marquardt}) and compute $\Delta\mathbf{\Theta}^n$ .
\item Compute $\rho_n$ as in~(\ref{eq:reduction_ratio}).
\item $\lambda^{n+1}\leftarrow 4\lambda^n$ if $\rho_n > \tfrac{3}{4}$; $\lambda^{n+1}\leftarrow \tfrac{1}{2}\lambda^n$ if $\rho_n < \tfrac{1}{4}$; otherwise, $\lambda^{n+1}\leftarrow \lambda^n$ .
\item $\mathbf{\Theta}^{n+1} \leftarrow \mathbf{\Theta}^n+ \Delta\mathbf{\Theta}^n$ if $\rho_n> 10^{-3}$; otherwise, $\mathbf{\Theta}^{n+1} \leftarrow \mathbf{\Theta}^n$.
\item $n=n+1$.
\end{enumerate}
\end{alg}

The backpropagation algorithm, which can be used for computing the
derivative of static functions is explained in the sequence.
How to adapt this algorithm in order to compute the Jacobian matrix both for parallel
and series-parallel training is explained in
Section~\ref{sec:nntraining}.

\subsection{Backpropagation}
\label{sec:modif-backpr}

Consider a multi-layer feedforward network, such as the three-layer network in
Figure~\ref{fig:network}. This network can be seen as a function that relates the
input $\mathbf{x}\in\mathbb{R}^{N_x}$ to the output $\mathbf{z}\in\mathbb{R}^{N_z}$.
The parameter vector $\mathbf{\Theta}$ contains
all weights $w^{(n)}_{i, j}$ and bias terms $\gamma^{(n)}_i$ of the network.
This subsection presents a modified version of backpropagation~\cite{hagan_training_1994}
for computing the neural
network output $\mathbf{z}$ and its Jacobian matrix
for a given input $\mathbf{x}$.
The notation used is the one displayed
in Figure~\ref{fig:network}.

\begin{figure}[htpb]
  \centering
  \includegraphics[width=1\textwidth]{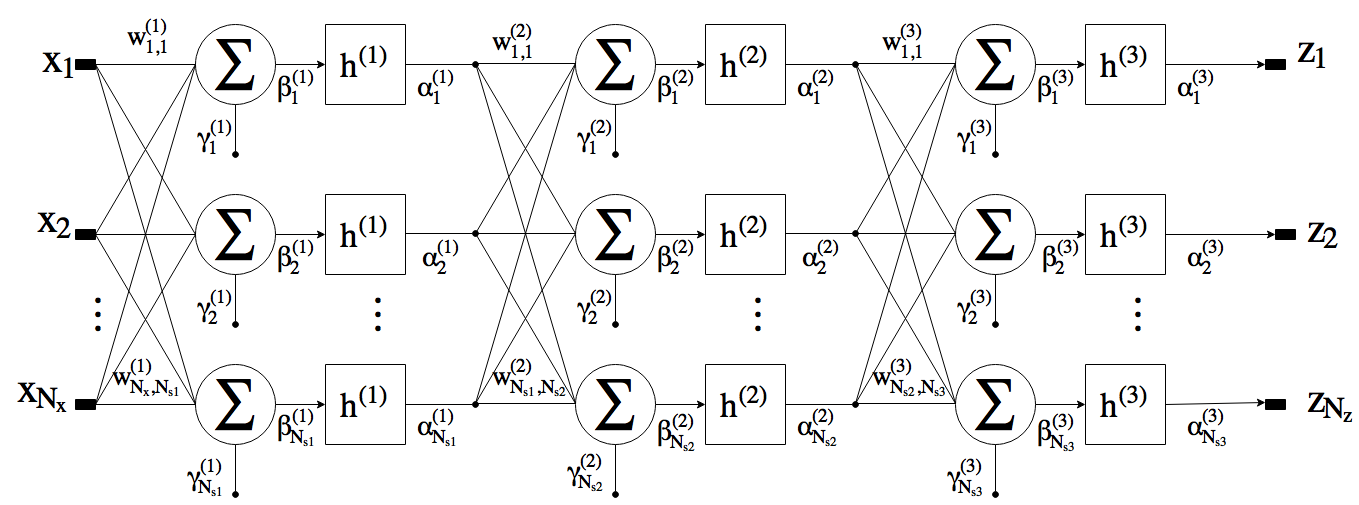}
  \caption{Three-layer feedforward network.}
  \label{fig:network}
\end{figure}

\subsubsection{Forward stage}
\label{sec:ForwardStage}

For a network with $\mathcal{L}$ layers the output nodes can be computed using
the following recursive matrix relation:
\begin{equation}
   \label{eq:neuralnet_forward}
  \boldsymbol{\alpha}^{(n)}
  =
  \begin{cases}
    \mathbf{x} & n = 0; \\
    \mathbf{h}^{(n)}(W^{(n)}\boldsymbol{\alpha}^{(n-1)} + \boldsymbol{\gamma}^{(n)})
    &  n = 1, \ldots, \mathcal{L},
    \end{cases}
\end{equation}

\noindent
where, for the $n$-th layer, $W^{(n)}$ is a matrix containing the weights $w^{(n)}_{i, j}$, 
$\boldsymbol{\gamma}^{(n)}$ is a vector containing the the bias terms $\gamma^{(n)}_i$
and $\mathbf{h}^{(n)}$ applies the nonlinear function $h^{(n)}$ element-wise.
The output $\mathbf{z}$ is given by:
\begin{equation}
  \label{eq:neuralnet_output}
  \mathbf{z}  =  \boldsymbol{\alpha}^{(\mathcal{L})}.
\end{equation}

\subsubsection{Backward stage}
\label{sec:BackwardStage}

The follow recurrence relation can be used to compute
$\tfrac{\partial  \mathbf{z}}{\partial \boldsymbol{\beta}^{(n)}}$
for every $n$:
\begin{equation}
  \label{eq:neuralnet_backward}
  \frac{\partial  \mathbf{z}}{\partial \boldsymbol{\beta}^{(n)}}
  =
  \begin{cases}
    \dot{H}^{(\mathcal{L})}\big(\boldsymbol{\beta}^{(\mathcal{L})}\big)&
    n = \mathcal{L}; \\
    \tfrac{\partial  \mathbf{z}}{\partial \boldsymbol{\beta}^{(n+1)}}\cdot
    W^{(n+1)}\cdot\dot{H}^{(n)}\big(\boldsymbol{\beta}^{(n)}\big) &
    n = \mathcal{L}-1,\ldots, 1.
  \end{cases}
\end{equation}

\noindent
where $\dot{H}^{(n)}$ is given by the following diagonal matrix:
\begin{equation*}
  \dot{H}^{(n)}\big(\boldsymbol{\beta}^{(n)}\big)
  = \text{diag}\Big(\dot{h}^{(n)}\big(\boldsymbol{\beta}^{(n)}_{1}\big),
  \cdots, \dot{h}^{(n)}\big(\boldsymbol{\beta}^{(n)}_{N_{sk}}\big)\Big).
\end{equation*}

The recursive expression~(\ref{eq:neuralnet_backward}) follows from
applying the chain rule\\
$\left({\tfrac{\partial  \mathbf{z}}{\partial \boldsymbol{\beta}^{(n)}}
    =
    \tfrac{\partial  \mathbf{z}}{\partial \boldsymbol{\beta}^{(n+1)}}
    \tfrac{\partial \boldsymbol{\beta}^{(n+1)}}{\partial \boldsymbol{\alpha}^{(n)}}
    \tfrac{\partial \boldsymbol{\alpha}^{(n)}}{\partial \boldsymbol{\beta}^{(n)}}}\right)$,
and considering
${\tfrac{\partial \boldsymbol{\beta}^{(n+1)}}{\partial \boldsymbol{\alpha}^{(n)}}=W^{(n+1)}}$
and
${\tfrac{\partial \boldsymbol{\alpha}^{(n)}}{\partial
  \boldsymbol{\beta}^{(n)}} = \dot{H}^{(n)}}$.

\subsubsection{Computing derivatives}

The derivatives of $\mathbf{z}$ with respect
to  $w^{(n)}_{i, j}$ and $\gamma^{(n)}_i$
can be used to form the Jacobian matrix
$\tfrac{\partial  \mathbf{z}}{\partial \mathbf{\Theta}}$
and  can be computed using the following expressions:
\begin{eqnarray}
  \label{eq:neuralnet_weight}
  \tfrac{\partial  \mathbf{z}}{\partial w^{(n)}_{i, j}}
  &=&
  \tfrac{\partial  \mathbf{z}}{\partial \beta^{(n)}_i}
  \tfrac{\partial \beta^{(n)}_i}{\partial w^{(n)}_{i, j}}
  =
  \tfrac{\partial  \mathbf{z}}{\partial \beta^{(n)}_i}
  \alpha^{(n-1)}_j;
  \\
  \label{eq:neuralnet_bias}
  \tfrac{\partial  \mathbf{z}}{\partial \gamma_i^{(n)}}
  &=&
  \tfrac{\partial  \mathbf{z}}{\partial \beta^{(n)}_i}
  \tfrac{\partial \beta^{(n)}_i}{\partial \gamma_i^{(n)}}
  =
  \tfrac{\partial  \mathbf{z}}{\partial \beta^{(n)}_i}.
\end{eqnarray}

\noindent
Furthermore, the derivatives of $\mathbf{z}$
with respect to the inputs $\mathbf{x}$ are:
\begin{equation}
  \label{eq:neuralnet_inputoutput_derivatives}
  \tfrac{\partial  \mathbf{z}}{\partial \mathbf{x}}
=
\tfrac{\partial  \mathbf{z}}{\partial \boldsymbol{\beta}^{(1)}}
\tfrac{\partial  \boldsymbol{\beta}^{(1)}}{\partial \boldsymbol{\alpha}^{(0)}}
=
\tfrac{\partial  \mathbf{z}}{\partial \boldsymbol{\beta}^{(1)}}
W^{(1)}.
\end{equation}

The backpropagation algorithm presented here
can be directly applied to series-parallel
training. For parallel training, however, a different procedure is needed.
A recurrent formula for that is introduced in the following section.

\section{Training}
\label{sec:nntraining}

Unlike other machine learning
applications (e.g. natural language processing and computer vision)
where there are enormous datasets available to train neural network
models, the datasets available for \textit{system identification}
are usually of moderate size.  The available data is usually obtained through tests with
limited duration because of practical
and economical reasons. And, even when there is a long record of input-output data,
it either does not contain meaningful dynamic
behavior~\cite{marquardt_algorithm_1963}
or the system cannot be considered time-invariant over the entire record,
resulting in the necessity of selecting smaller portions of this longer
dataset for training.

Because large datasets are seldom available in system identification problems,
neural networks for such applications are usually restricted to a few hundred weights.
The Levenberg-Marquardt method does provide a
fast convergence rate~\cite{nocedal_numerical_2006}
and has been described as very efficient
for batch training of moderate size problems~\cite{hagan_training_1994},
where the memory used by this algorithm is not prohibitive.
Hence, it will be the method of choice for
training neural networks in this paper.

Besides that, recurrent neural networks often present vanishing gradients
that may prevent the progress of the optimization algorithm.
The use of second-order information (as in the Levenberg-Marquardt algorithm)
helps to mitigate this problem~\cite{bengio_learning_1994}.

In the series-parallel configuration the parameters are estimated by minimizing
 $\tfrac{1}{2}\|\mathbf{e}_1\|^2$, what can be done using the algorithm
described in Section~\ref{sec:non-linear-least-squares}. The required Jacobian matrix
$\tfrac{\partial \mathbf{e}_1}{\partial \mathbf{\Theta}}$ can be computed
according to the following well known result.
\begin{prop}
  \label{prop:seri-parall-train-1}
  The Jacobian matrix of $\mathbf{e}_1$ with respect to
  $\mathbf{\Theta}$  is \\
  $\tfrac{\partial \mathbf{e}_1}{\partial \mathbf{\Theta}} = \left[\tfrac{\partial \mathbf{e}_1[1]}{\partial \mathbf{\Theta}}^T, \cdots,
  \tfrac{\partial \mathbf{e}_1[N]}{\partial \mathbf{\Theta}}^T\right]^T$,
  where
  ${\tfrac{\partial  \mathbf{e}_1[k]}{\partial \mathbf{\Theta}} = 
     \tfrac{\partial  \hat{\mathbf{y}}_1[k]}{\partial \mathbf{\Theta}} =
    \tfrac{\partial \mathbf{F}}{\partial \mathbf{\Theta}}(\underline{\mathbf{y}}_{[k]},
    \underline{\mathbf{u}}_{[k]}; \mathbf{\Theta})}$
 that can be computed using the backpropagation described
 in Section~\ref{sec:modif-backpr}.
\end{prop}

\begin{proof}
  This results readily from differentiating~(\ref{eq:one-step-ahead-prediction}).
\end{proof}

In the parallel configuration the parameters are estimated by minimizing
$\tfrac{1}{2}\|\mathbf{e}_s\|^2$. There are
two different ways to take into account the initial conditions
${\underline{\mathbf{y}}_0}$: (i)~by fixing
$\underline{\mathbf{y}}_0$ and estimating  $\mathbf{\Theta}$;
and, (ii)~by defining an extended parameter vector
${\mathbf{\Phi} = [\mathbf{\Theta}^T,\underline{\mathbf{y}}_0^T]^T}$ 
and estimating  $\underline{\mathbf{y}}_0$  and $\mathbf{\Theta}$ simultaneously.

When using formulation (i), a suitable choice
is to set the initial conditions equal to the measured
outputs (${\mathbf{y}_0[k] = \mathbf{y}[k],~ k = 1, \ldots, n_y}$).
When using formulation (ii) the measured output vector may be used as
an initial guess that will be refined by the optimization algorithm.

The optimal choice for the initial condition would be
${\mathbf{y}_0[k]=\mathbf{y}^*[k]}$ for \\
${k=1, \cdots, n_y}$. Formulation (i) uses the non-optimal choice
${\mathbf{y}_0[k] = \mathbf{y}[k]\neq \mathbf{y}^*[k]}$.
Formulation (ii) goes one step further and include the initial
conditions $\mathbf{y}_0[k]$ in the optimization problem,
so it converges to $\mathbf{y}^*[k]$ and hence improves the
parameter estimation.

The required Jacobian matrices
$\tfrac{\partial \mathbf{e}_s}{\partial \mathbf{\Theta}}$ and
$\tfrac{\partial \mathbf{e}_s}{\partial \underline{\mathbf{y}}_0}$
can be computed according to the following proposition.
\begin{prop}
  \label{prop:main}
  The Jacobian matrices of $\mathbf{e}_1$ with respect to
  $\mathbf{\Theta}$  and $\underline{\mathbf{y}}_0$ are
  ${\tfrac{\partial \mathbf{e}_s}{\partial \mathbf{\Theta}} =
  \left[\tfrac{\partial \mathbf{e}_s[1]}{\partial \mathbf{\Theta}}^T, \cdots,
    \tfrac{\partial \mathbf{e}_s[N]}{\partial \mathbf{\Theta}}^T\right]^T}$
  and
  ${\tfrac{\partial \mathbf{e}_s}{\partial \underline{\mathbf{y}}_0} =
  \left[\tfrac{\partial \mathbf{e}_s[1]}{\partial \underline{\mathbf{y}}_0}^T, \cdots,
    \tfrac{\partial \mathbf{e}_s[N]}{\partial \underline{\mathbf{y}}_0}^T\right]^T}$
  where
  ${\tfrac{\partial  \mathbf{e}_s[k]}{\partial \mathbf{\Theta}} = 
    \tfrac{\partial  \hat{\mathbf{y}}_s[k]}{\partial \mathbf{\Theta}}}$
  and
  ${\tfrac{\partial  \mathbf{e}_s[k]}{\partial \underline{\mathbf{y}}_0} = 
    \tfrac{\partial  \hat{\mathbf{y}}_s[k]}{\partial \underline{\mathbf{y}}_0}}$
  can be computed according to the following recursive formulas:
  \begin{small}
    \begin{equation}
      \label{eq:jacobian_recursive}
      \frac{\partial\hat{\mathbf{y}}_s[k]}{\partial \mathbf{\Theta}}
      = 
      \begin{dcases}
        \mathbf{0},~ 1 \le k \le n_y; \\
        \begin{split}
        \tfrac{\partial\mathbf{F}}{\partial \mathbf{\Theta}}(\hat{\mathbf{y}}_{[k]},
        \mathbf{u};_{[k]} \mathbf{\Theta})
          +\sum_{i=1}^{n_y}
          \tfrac{\partial\mathbf{F}}{\partial\mathbf{y}[k-i]}(\hat{\mathbf{y}}_{[k]},
          \mathbf{u}_{[k]};& \mathbf{\Theta})
          \tfrac{\partial\hat{\mathbf{y}}_s[k-i]}{\partial \mathbf{\Theta}}, \\
          & k > n_y,
        \end{split}
      \end{dcases}
    \end{equation}
  \end{small}
  \begin{small}
    \begin{equation}
      \label{eq:jacobian_recursive_init}
      \frac{\partial\hat{\mathbf{y}}_s[k]}{\partial \underline{\mathbf{y}}_0}
      = 
      \begin{dcases}
        D^{(n)},~ 1 \le k \le n_y;\\
        \begin{split}
        \sum_{i=1}^{n_y}&
        \tfrac{\partial\mathbf{F}}{\partial\mathbf{y}[k-i]}(\hat{\mathbf{y}}_{[k]},
        \mathbf{u}_{[k]}; \mathbf{\Theta})
        \tfrac{\partial\hat{\mathbf{y}}_s[k-i]}{\partial
          \underline{\mathbf{y}}_0},~ k > n_y,
        \end{split}
      \end{dcases}
    \end{equation}
  \end{small}
  where $D^{(n)}\in \mathbb{R}^{N_y \times n_yN_y}$ is defined as:
  \begin{equation}
    \big\{D^{(n)}\big\}_{i,j} =
    \begin{dcases}
      1, & \text{if } j=(n-1)\cdot N_y+i, \\
      0, & \text{otherwise.}
    \end{dcases}
    \end{equation}
  \end{prop}

  \begin{proof}
    The proof follows from differentiating~(\ref{eq:free_run_simulation_recursion})
    and applying the chain rule.
  \end{proof}

  This formula can be interpreted either as a variation of the dynamic
backpropagation~\cite{narendra_identification_1990} adapted to compute the Jacobian
instead of the gradient; or,  as a specific case of real-time recurrent
learning~\cite{williams_experimental_1989} with a special type of recurrent connection.

\section{Complexity analysis}
\label{sec:complexity-analysis}

In this section we present a novel complexity analysis  comparing
series-parallel  training (SP), parallel training with fixed initial conditions
(P$\mathbf{\Theta}$) and parallel training with extended parameter vector
(P$\mathbf{\Phi}$).
We show that the training methods have similar computational cost for the
nonlinear least-squares formulation.
The number of floating point operations (\textit{flops})
is estimated based on \cite[Table 1.1.2]{golub_matrix_2012}. Low-order terms,
as usual, are neglected in the analysis.

\subsection{Neural network output and its partial derivatives}

The backpropagation algorithm described in Section~\ref{sec:modif-backpr}
can be used for training both fully or partially connected networks.
The differences relate to the internal representation of the weight
matrices $W^{(n)}$: for a partially connected network the matrices
are stored using a sparse representation, e.g.\ compressed sparse
column (CSC) representation.

The total number of \textit{flops} required to evaluate the output and to compute 
partial derivatives for a fully connected feedforward network is summarized in
Table~\ref{tab:backprop_flops}. $N_{w}$ and $N_{\gamma}$ are respectively the total number
of weights and of bias terms, such that $N_w+N_\gamma = N_\Theta$.
For this fully connected network:
\begin{eqnarray}
  \nonumber
  N_{w} &=& N_{x}\cdot N_{s1}+N_{s1}\cdot N_{s2}+\cdots+N_{s(\mathcal{L}-1)}\cdot
            N_{s\mathcal{L}}, \\
  \label{eq:Ngamma}
  N_{\gamma} &=&  N_{s1}+N_{s2}+\cdots+N_{s\mathcal{L}}.
\end{eqnarray}
\begin{table}[htpb]
  \setlength\extrarowheight{2pt}
  \centering
  \caption{Modified backpropagation number of \textit{flops} for a fully connected network.}
  \begin{tabular}{|l|m{4cm}|}
    \hline
    \multicolumn{2}{|c|}{Computing Neural Network Output} \\ \hline
    i) Compute  $\mathbf{F(\mathbf{x}; \mathbf{\Theta})}$ ---
    Eq.~(\ref{eq:neuralnet_forward})-(\ref{eq:neuralnet_output})& $2N_w$  \\ \hline \hline
    \multicolumn{2}{|c|}{Computing Partial Derivatives} \\ \hline
    ii) Backward Stage --- Eq.~(\ref{eq:neuralnet_backward})
    & $(2 N_z+ 1)(N_w- N_x N_{s1})$ \\ \hline
    iii) Compute
    $\tfrac{\partial \mathbf{F}}{\partial \mathbf{\Theta}}$
    --- Eq.~(\ref{eq:neuralnet_weight})-(\ref{eq:neuralnet_bias})
    &$N_w\cdot N_z$ \\[2pt] \hline
    iv) Compute
    $\tfrac{\partial \mathbf{F}}{\partial \mathbf{x}}$
    --- Eq.~(\ref{eq:neuralnet_inputoutput_derivatives})
   &$2 N_x\cdot N_{s1} \cdot N_z$ \\[2pt] \hline
  \end{tabular}
  \label{tab:backprop_flops}
\end{table}
Since the more relevant terms of the complexity analysis
in Table~\ref{tab:backprop_flops}
are being expressed in terms of the number of weights $N_w$  the results for a fully connected network are
similar to the ones that would be obtained for a partially
connected network using a sparse representation.

\subsection{Number of flops for series-parallel  and parallel training}

\begin{table}[htpb]
  \setlength\extrarowheight{2pt}
  \centering
   \caption{Levenberg-Marquardt number of \textit{flops} per iteration for series-parallel
     training (SP), parallel training with fixed initial conditions
     (P$\mathbf{\Theta}$) and parallel training with extended parameter vector
     (P$\mathbf{\Phi}$). A mark \xmark\ signals which
    calculation is required in each method.}
  \begin{tabular}{|l|m{3.7cm}|c|c|c|}
    \cline{3-5}
    \multicolumn{2}{c|}{} & SP
    & P$\mathbf{\Theta}$
    & P$\mathbf{\Phi}$ \\ \hline
    \multicolumn{2}{|c|}{Computing Error} & & &  \\ \cline{1-2}
    i) Compute  $\mathbf{F(\mathbf{x}; \mathbf{\Theta})}$ & $2N\cdot N_w$
    & \xmark & \xmark & \xmark \\ \hline
    \hline
    \multicolumn{2}{|c|}{Computing  Partial Derivatives} & & & \\  \cline{1-2}
    ii) Backward Stage
    & {\scriptsize $N(2 N_y+ 1)(N_w- N_x N_{s1})$}
    & \xmark & \xmark & \xmark \\ \cline{1-2}
    iii) Compute
    $\tfrac{\partial \mathbf{F}}{\partial \mathbf{\Theta}}$
    & $N\cdot N_w\cdot N_y$
    & \xmark & \xmark & \xmark \\[2pt] \cline{1-2}
    iv) Compute
    $\tfrac{\partial \mathbf{F}}{\partial \mathbf{x}}$
    & $2 N\cdot N_x \cdot N_{s1}\cdot N_y$
    & & \xmark &\xmark  \\[2pt] \cline{1-2}
    v) Equation~(\ref{eq:jacobian_recursive})
    & $2 N\cdot  N_\Theta \cdot (N_y^2 +N_y)$
    &  & \xmark &\xmark \\ \cline{1-2}
    vi) Equation~(\ref{eq:jacobian_recursive_init})
    & $2 N\cdot n_y \cdot N_y\cdot(N_y^2+N_y) $
    &  & & \xmark\\ \hline
    \hline
    \multicolumn{2}{|c|}{Solving Equation~(\ref{eq:levenberg_marquardt})}  & & &\\ \cline{1-2}
    vii) Solve~(\ref{eq:levenberg_marquardt}) ---  $\mathbf{\Theta}$
    & $2  N\cdot N_\Theta^2 + \tfrac{1}{3} N_\Theta^3 $
    & \xmark & \xmark & \\[2pt] \cline{1-2}
    viii) Solve~(\ref{eq:levenberg_marquardt}) ---  $\mathbf{\Phi}$
    & $2 N\cdot N_\Phi^2 + \tfrac{1}{3} N_\Phi^3 $
    & & & \xmark\\[2pt] \hline
  \end{tabular}
  \label{tab:error_flops}
\end{table}

The number of flops of each iteration of the Levenberg-Marquardt
algorithm  is summarized in Table~\ref{tab:error_flops}.
Entries (i) to (iv) in Table~\ref{tab:error_flops} follow directly
from Table~\ref{tab:backprop_flops}, considering  ${N_z = N_y}$ 
and multiplying the costs by $N$ because of the number of different
inputs being evaluated. Furthermore, in entries (v) and (vi)
the  evaluation of~(\ref{eq:jacobian_recursive})
and~(\ref{eq:jacobian_recursive_init}) is accomplished by storing computed values
and performing only one new matrix-matrix product per evaluation.

The cost of solving~(\ref{eq:levenberg_marquardt})
is about $2\cdot N\cdot N_\Theta^2 + \tfrac{1}{3}\cdot N_\Theta^3 $
where the cost $2N\cdot N_\Theta^2$ is due to the multiplication of
the Jacobian matrix by its transpose and $\tfrac{1}{3}N_\Theta^3$ is
due to the needed Cholesky factorization. When using an extended parameter
vector, $N_\Theta$ is replaced with $N_\Phi$ in the analysis.

\subsection{Comparing methods}

Assuming the number of nodes in the last hidden layer is greater than the
number of outputs ($N_y < N_{s(\mathcal{L}-1)}$), the inequalities apply:
\begin{eqnarray}
  n_y\cdot N_y \le N_x < N_x\cdot N_{s1} \le N_w < N_\Theta; \label{eq:relations} \\
  N_y < N_y^2 < N_y \cdot N_{s(\mathcal{L}-1)} \le N_w < N_\Theta. \nonumber
\end{eqnarray}

\noindent
From Table~\ref{tab:error_flops} and from the above inequalities
it follows that the cost of each Levenberg-Marquardt iteration is
dominated by the cost of solving Equation~(\ref{eq:levenberg_marquardt}).
Furthermore, ${N_\Phi=N_\Theta+n_y\cdot N_y < 2 N_\Theta}$,
hence, the asymptotic computational cost of the training method S$\mathbf{\Phi}$ is
the same as that of SP and S$\mathbf{\Theta}$
methods: ${\mathcal{O}(N \cdot N_\Theta^2 + N_\Theta^3)}$.

From Table~\ref{tab:error_flops} it is also possible to analyze
the cost of each of the major steps needed in each full iteration
of the algorithm:

\begin{itemize}
\item \textbf{Computing Error:}
  The cost of computing the error is the same for all of the training methods.
\item \textbf{Computing Partial Derivatives:}
The  computation of  partial derivatives has a cost of
$\mathcal{O}(N \cdot N_w \cdot N_y)$
for the SP training method and a cost of $\mathcal{O}(N\cdot N_\Theta \cdot N_y^2)$
for both P$\mathbf{\Phi}$ and P$\mathbf{\Theta}$. For many cases of interest
in system identification, the number of model outputs $N_y$ is small. Furthermore,
${N_\Theta = N_w + N_\gamma \approx N_w}$ (see Eq.~(\ref{eq:Ngamma})).
That is why the cost of computing the partial derivatives
for parallel training is comparable to the cost for series-parallel training.
\item \textbf{Solving Equation~(\ref{eq:levenberg_marquardt}):}
It already has been established that the cost of this
step $\mathcal{O}(N \cdot N_\Theta^2 + N_\Theta^3)$ dominates the computational
cost for all the training methods. Furthermore $n_y\cdot N_y$ is usually much
smaller than $N_\Theta$ such that $N_\Phi \approx N_\Theta$ and the number of flops of this
stage is basically the same for all the training methods.
\end{itemize}

\subsection{Memory complexity}

Considering that $N\gg n_y$ , it follows that, for the
three training methods, the storage capacity is dominated
by the storage of the Jacobian matrix
${\tfrac{\partial \mathbf{e}}{\partial \mathbf{\Theta}}}$
or of the matrix resulting from the product
${\left[\tfrac{\partial \mathbf{e}}{\partial \mathbf{\Theta}}\right]^T
  \left[\tfrac{\partial \mathbf{e}}{\partial \mathbf{\Theta}}\right]}$.
Therefore the memory size required by the algorithm is
about $\mathcal{O}(\max(N\cdot N_y \cdot N_\Theta,~N_\Theta^2))$.

For very large datasets and a large number of parameters,
this storage requirement may be prohibitive
and others methods should be used (e.g. stochastic gradient descent or variations).
Nevertheless, for datasets of moderate size and networks with few hundred parameters,
as it is usually the case for system identification,
the use of nonlinear least-squares is a viable option.

\section{Unifying framework}
\label{sec:unifying-framework}

In this section we present parallel and series-parallel training
in the \textit{prediction error methods} framework~\cite{ljung_convergence_1978,
  ljung_system_1998}. This analysis provides some insight about the situations
in which one training method outperforms the other one.

\subsection{Output error \textit{vs} equation error}
\label{sec:output-vs-equation-error}

To study the previously described problem it is assumed that for a given input
sequence ${\mathbf{u}[k]}$ and a set of initial conditions ${\mathbf{y}^*_0}$
the output was generated by a \textit{``true system''},
described by the following equations:
  \begin{align}
    \mathbf{y}^*[k] =& \mathbf{F}^*(\mathbf{y}^*[k-1], \hdots, \mathbf{y}^*[k-n_y],
                       \mathbf{u}[k-\tau_d], \hdots,
                       \mathbf{u}[k-n_u]; \mathbf{\Theta}^*) + \mathbf{v}[k] \nonumber\\
    \mathbf{y}[k] =& \mathbf{y}^*[k] + \mathbf{w}[k] , \label{eq:true_system}
  \end{align}
where  $\mathbf{F}^*$ and $\mathbf{\Theta}^*$ are
the \textit{``true''} function and parameter vector that describe the
system;
$\mathbf{v}[k]\in \mathbb{R}^{N_y}$ and
${\mathbf{w}[k]\in \mathbb{R}^{N_y}}$ are random variable vectors,
that cause the deviation of the deterministic model from its true value;
$\mathbf{u}[k]$ and $\mathbf{y}[k]$ are the measured input and output vectors;
and, $\mathbf{y}^*[k]$ is the output vector without the effect of the output error.

The random variable $\mathbf{v}[k]$ affects the system
dynamics and is called \textit{equation error}, while the random variable
$\mathbf{w}[k]$ only affects the measured values and is called
\textit{output error}.

\subsection{Optimal predictor}
\label{sec:optimal_predictor}

If the measured values of $\mathbf{y}$ and $\mathbf{u}$ are known at all instants
previous to $k$, the optimal prediction
of $\mathbf{y}[k]$ is the following conditional expectation:
\footnote{The prediction is optimal
  in the sense that the expected squared prediction error is
  minimized~\cite[p.18, Sec. 2.4]{friedman_elements_2001}.}
\begin{equation}
  \label{eq:conditional_expectation}
  \hat{\mathbf{y}}_*[k] = E\left\{\mathbf{y}[k]~\Big\vert~\underline{\mathbf{y}}_{[k]},
                       \underline{\mathbf{u}}_{[k]} \right\}
\end{equation}
where  $\hat{\mathbf{y}}_*[k]$  denotes the optimal prediction
and $E\{\cdot\}$ indicates the mathematical expectation.

Consider the following situations:
\begin{situation}[White equation error]
  \label{thm:white-equation-error}
  The sequence of equation errors $\{\mathbf{v}[k]\}$
  is a white noise process and the output error
  is zero $({\mathbf{w}[k]=0})$.
\end{situation}
\begin{situation}[White output error]
  \label{thm:white-output-error}
  The sequence of output errors $\{\mathbf{w}[k]\}$
  is a white noise process and the equation error
  is zero $(\mathbf{v}[k]=0)$.
\end{situation}

The next two lemmas give the optimal prediction $\hat{\mathbf{y}}_*[k]$
for the two situations above:

\begin{lemma}
  \label{lemma:series-parallel}
  If Situation \ref{thm:white-equation-error} holds
  and the function and parameter vector matches the true ones
  $(\mathbf{F} = \mathbf{F}^*$ and $\mathbf{\Theta} = \mathbf{\Theta}^*)$,
  then the one-step-ahead prediction
  is equal to the optimal prediction $(\hat{\mathbf{y}}_1[k] = \hat{\mathbf{y}}_*[k])$.
\end{lemma}

\begin{proof}
  Since the output error is zero, it follows that $\mathbf{y}[k]=\mathbf{y}^*[k]$
  and therefore Equation~(\ref{eq:true_system}) reduces to
  $\mathbf{y}[k] = \mathbf{F}^*(\underline{\mathbf{y}}_{[k]}, \underline{\mathbf{u}}_{[k]};
  \mathbf{\Theta}^*)+\mathbf{v}[k]$.

  And, because $\mathbf{v}[k]$ has zero mean\footnote{A white noise process has zero mean by definition.}, it follows that:
  \begin{equation*}
    \hat{\mathbf{y}}_*[k] = E\{\mathbf{y}[k]\,\mid\, \underline{\mathbf{y}}_{[k]},
                       \underline{\mathbf{u}}_{[k]}\} =
    \mathbf{F}^*(\underline{\mathbf{y}}_{[k]}, \underline{\mathbf{u}}_{[k]};
    \mathbf{\Theta}^*)= \hat{\mathbf{y}}_1[k].
  \end{equation*}
\end{proof}

\begin{lemma}
  \label{lemma:parallel}
  If Situation \ref{thm:white-output-error} holds,
  and the function, parameter vector and initial conditions matches the true ones
  $(\mathbf{F} = \mathbf{F}^*$, $\mathbf{\Theta} = \mathbf{\Theta}^*$
  and $\mathbf{y}_0 = \mathbf{y}^*_0)$,
  then the free-run simulation 
  is the optimal prediction $(\hat{\mathbf{y}}_s[k] = \hat{\mathbf{y}}_*[k])$.
\end{lemma}

\begin{proof}
 There is no equation error and therefore:
 \begin{equation*}
   \hat{\mathbf{y}}_*[k] =
   E\{\mathbf{y}[k]\, \mid\,\underline{\mathbf{y}}_{[k]},
   \underline{\mathbf{u}}_{[k]}\} =
   E\{\mathbf{y}^*[k]+\mathbf{w}[k]\, \mid\,\underline{\mathbf{y}}_{[k]},
         \underline{\mathbf{u}}_{[k]}\}
   = \mathbf{y}^*[k] =
   \hat{\mathbf{y}}_s[k],
 \end{equation*}
 where it was used that for matching initial conditions and parameters and in the absence
 of equation error, the noise-free output ${\mathbf{y}^*[k]}$
 is exactly equal to the free-run simulation
 (${\mathbf{y}^*[k] = \hat{\mathbf{y}}_s[k]}$).
\end{proof}

Hence, for $\mathbf{F} = \mathbf{F}^*$, both training methods
minimize an error that approaches the optimal predictor error
${\mathbf{e}_*[k]=\hat{\mathbf{y}}_*[k] - \mathbf{y}[k]}$ as
$\mathbf{\Theta}\rightarrow\mathbf{\Theta}^*$.
Series-parallel training does it for Situation~\ref{thm:white-equation-error},
and parallel training for Situation~\ref{thm:white-output-error}.
It follows from~\cite{ljung_convergence_1978} that, under additional assumptions,
series-parallel training is a consistent estimator for Situation~\ref{thm:white-equation-error}
and parallel training is a consistent estimator for Situation~\ref{thm:white-output-error}.


\section{Implementation and test results}
\label{sec:examples}

The implementation is in Julia and runs on a computer
with a processor Intel(R) Core(TM) i7-4790K CPU @ 4.00GHz.
For all examples in this paper, the activation function
used in hidden layers is the hyperbolic tangent, the initial values
of the weights $w_{i,j}^{(n)}$ are drawn from a zero mean
normal distribution with standard
deviation ${\sigma=(N_{s (n)})^{-0.5}}$  and the bias terms
$\gamma_i^{(n)}$ are initialized with zeros~\cite{lecun_efficient_2012}.
Also, in all parallel training examples the parameter vector is
extended with the initial conditions for the optimization process (P$\Phi$ training).

The free-run  mean-square simulation error (${\text{MSE} = \tfrac{1}{N}\sum_{n=1}^N (y[k] - \hat{y}[k])^2}$)
is used to compare the models over the validation window.

The first example compares the training method using data from an experimental plant and
the second one investigates different noise configurations on computer generated data.
The code and data used in the numerical examples
are available.\footnote{GitHub repository: 
\href{https://github.com/antonior92/ParallelTrainingNN.jl}{https://github.com/antonior92/ParallelTrainingNN.jl}.}

\subsection{Example 1: Data from a pilot plant}

In this example, input-output signals were collected
from \textit{LabVolt Level Process Station}
(model \textit{3503-MO}~\cite{labvolt_mobile_2015}).
This plant consists of a tank that is filled with water
from a reservoir. The water is pumped  at fixed speed, while the flow is regulated
using a pneumatically operated control valve driven by a voltage $u[k]$.
The water column height $y[k]$ is indirectly measured using a pressure
sensor at the bottom of the tank. Figure~\ref{fig:pilotplant_val}
shows the free-run simulation over the validation window of models
obtained for this process using parallel
and series-parallel training.

\begin{figure}[ht]
  \centering
  \begin{tikzpicture}
  \begin{axis}[
    yscale=0.5,
    y label style={at={(0.01,1)}},
    width=1\textwidth,
    xlabel=$t$ (hours),
    ylabel=$y$ (cm),
    xmin=0,
    xmax=1,
    ymin=10,
    ymax=70,
    legend entries={$y$, SP, P},
    legend style={at={(0.7, 1.95)},
      anchor={north east}}]

    \addplot[
    color=black,
    no marks,
    line width=0.8pt]
    table [
    x expr=\thisrow{"time"}/3600,
    y expr=\thisrow{"output"}/10,
    col sep=comma] {./data/validation_results.csv};\

    \addplot[
    color={rgb,1:red,0.00000000;green,0.50196078;blue,0.00000000},
    no marks,
    line width=1.2pt,
    style=dashed]
    table [
    x expr=\thisrow{"time"}/3600,
    y expr=\thisrow{"SP_simulation"}/10,
    col sep=comma] {./data/validation_results.csv};\

    \addplot[
    color=blue,
    no marks,
    line width=1.2pt,
    style=dashed]
    table [
    x expr=\thisrow{"time"}/3600,
    y expr=\thisrow{"PPhi_simulation"}/10,
    col sep=comma] {./data/validation_results.csv};\
  \end{axis}
\end{tikzpicture}
  \caption{Validation window for pilot plant. Free-run simulation 
    for models obtained using series-parallel (SP) and parallel (P) training.
    The mean square errors are
    ${\text{MSE}_{\text{SP}} = 1144.6}$;
    ${\text{MSE}_{\text{P}} = 296.2}$.
    The models have ${n_y=n_u=1}$ and 10 nodes in the hidden layer  and
    were trained on a two hour long  dataset sampled at $T_s=10{\rm s}$.
    The same initial parameter guess was used for both training methods.
    The training was 100 epochs long, which took 3.3 and 3.9 s,
    respectively, for series-parallel and parallel training.}
  \label{fig:pilotplant_val}
\end{figure}
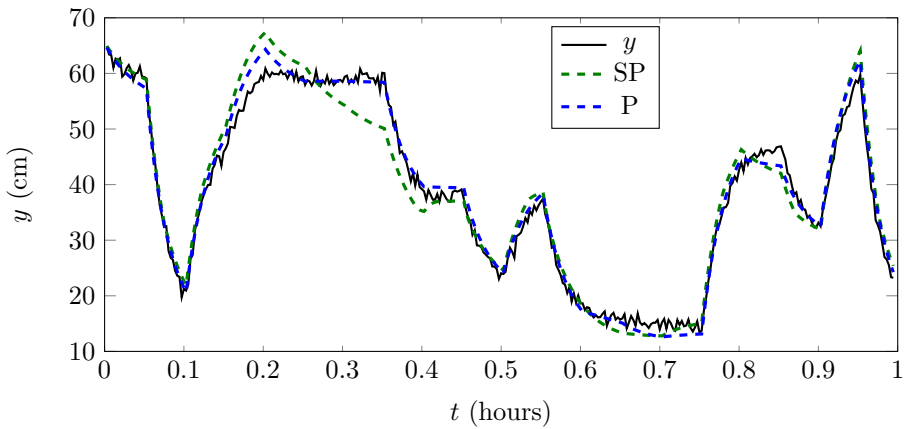

Since parameters are randomly initialized,
different realizations will yield different results.
Figure~\ref{fig:pilotplant_boxplots} shows the validation
errors for both training methods for randomly
drawn initial guesses. While parallel training
consistently provides models with better validation results than
series-parallel training, it also has some outliers that result in very poor
validation results. Such outliers probably happen as the algorithm
gets trapped in ``poor'' local minima during parallel training.

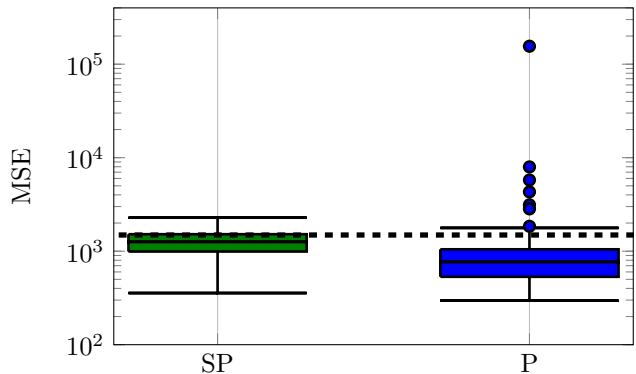
\begin{figure}[ht]
  \centering
  \begin{tikzpicture}[]
  \begin{axis}[height = {0.5\textwidth}, ylabel = {MSE}, xmin = {0.03399999999999999}, xmax = {2.3659999999999997},
    ymax = {400000}, ymode = {log}, xlabel = {}, unbounded coords=jump, xticklabel style={rotate = 0},
    xmajorgrids = true, xtick = {0.5,1.9}, xticklabels = {SP,P}, yticklabel style={rotate = 0},
    log basis y=10, ymajorgrids = false,     xshift = 0.0mm,
    yshift = 0.0mm,
    axis background/.style={fill={rgb,1:red,1.00000000;green,1.00000000;blue,1.00000000}},
    ymin = {100}, width = {0.7\textwidth}]
\addplot+ [only marks = {true}, color = {rgb,1:red,0.00000000;green,0.00000000;blue,0.00000000},
draw opacity=1.0,
line width=1,
solid,mark = *,
mark size = 2.0,
mark options = {
    color = {rgb,1:red,0.00000000;green,0.00000000;blue,0.00000000}, draw opacity = 1.0,
    fill = {rgb,1:red,0.00000000;green,0.50196078;blue,0.00000000}, fill opacity = 1.0,
    line width = 1,
    rotate = 0,
    solid
},forget plot]coordinates {
};
\addplot+ [color = {rgb,1:red,0.00000000;green,0.00000000;blue,0.00000000},
draw opacity=1.0,
line width=1,
solid,mark = none,
mark size = 2.0,
mark options = {
    color = {rgb,1:red,0.00000000;green,0.00000000;blue,0.00000000}, draw opacity = 1.0,
    fill = {rgb,1:red,0.00000000;green,0.50196078;blue,0.00000000}, fill opacity = 1.0,
    line width = 1,
    rotate = 0,
    solid
},fill = {rgb,1:red,0.00000000;green,0.50196078;blue,0.00000000}, fill opacity=1.0,forget plot]coordinates {
(0.5, 356.9096536818575)
(0.09999999999999998, 356.9096536818575)
(0.9, 356.9096536818575)
(0.5, 356.9096536818575)
(0.5, 991.6188960568412)
(NaN, NaN)
(0.09999999999999998, 991.6188960568412)
(0.09999999999999998, 1258.0036193507356)
(0.9, 1258.0036193507356)
(0.9, 991.6188960568412)
(0.09999999999999998, 991.6188960568412)
(NaN, NaN)
(0.09999999999999998, 1519.0391743052073)
(0.09999999999999998, 1258.0036193507356)
(0.9, 1258.0036193507356)
(0.9, 1519.0391743052073)
(0.09999999999999998, 1519.0391743052073)
(NaN, NaN)
(0.5, 2281.336834621979)
(0.09999999999999998, 2281.336834621979)
(0.9, 2281.336834621979)
(0.5, 2281.336834621979)(0.5, 1519.0391743052073)
};
\addplot+ [only marks = {true}, color = {rgb,1:red,0.00000000;green,0.00000000;blue,0.00000000},
draw opacity=1.0,
line width=1,
solid,mark = *,
mark size = 2.0,
mark options = {
    color = {rgb,1:red,0.00000000;green,0.00000000;blue,0.00000000}, draw opacity = 1.0,
    fill = {rgb,1:red,0.00000000;green,0.00000000;blue,1.00000000}, fill opacity = 1.0,
    line width = 1,
    rotate = 0,
    solid
},forget plot]coordinates {
(1.9, 5768.584922548905)
(1.9, 155567.435127202)
(1.9, 3133.4840843464826)
(1.9, 1846.127013549703)
(1.9, 4309.66519879118)
(1.9, 2834.8843527242443)
(1.9, 7962.27952190591)
};
\addplot+ [color = {rgb,1:red,0.00000000;green,0.00000000;blue,0.00000000},
draw opacity=1.0,
line width=1,
solid,mark = none,
mark size = 2.0,
mark options = {
    color = {rgb,1:red,0.00000000;green,0.00000000;blue,0.00000000}, draw opacity = 1.0,
    fill = {rgb,1:red,0.00000000;green,0.00000000;blue,1.00000000}, fill opacity = 1.0,
    line width = 1,
    rotate = 0,
    solid
},fill = {rgb,1:red,0.00000000;green,0.00000000;blue,1.00000000}, fill opacity=1.0,forget plot]coordinates {
(1.9, 296.2163307861891)
(1.5, 296.2163307861891)
(2.3, 296.2163307861891)
(1.9, 296.2163307861891)
(1.9, 532.3577912417494)
(NaN, NaN)
(1.5, 532.3577912417494)
(1.5, 768.3208182657944)
(2.3, 768.3208182657944)
(2.3, 532.3577912417494)
(1.5, 532.3577912417494)
(NaN, NaN)
(1.5, 1047.754190794732)
(1.5, 768.3208182657944)
(2.3, 768.3208182657944)
(2.3, 1047.754190794732)
(1.5, 1047.754190794732)
(NaN, NaN)
(1.9, 1773.956037027588)
(1.5, 1773.956037027588)
(2.3, 1773.956037027588)
(1.9, 1773.956037027588)
(1.9, 1047.754190794732)
};
\addplot+ [color = {rgb,1:red,0.00000000;green,0.00000000;blue,0.00000000},
draw opacity=1.0,
line width=2,
dashed,mark = none,
mark size = 2.0,
mark options = {
    color = {rgb,1:red,0.00000000;green,0.00000000;blue,0.00000000}, draw opacity = 1.0,
    fill = {rgb,1:red,0.00000000;green,0.00000000;blue,0.00000000}, fill opacity = 1.0,
    line width = 1,
    rotate = 0,
    solid
},forget plot]coordinates {
(-0.3, 1489.5)
(-0.2, 1489.5)
(-0.1, 1489.5)
(0.0, 1489.5)
(0.1, 1489.5)
(0.2, 1489.5)
(0.3, 1489.5)
(0.4, 1489.5)
(0.5, 1489.5)
(0.6, 1489.5)
(0.7, 1489.5)
(0.8, 1489.5)
(0.9, 1489.5)
(1.0, 1489.5)
(1.1, 1489.5)
(1.2, 1489.5)
(1.3, 1489.5)
(1.4, 1489.5)
(1.5, 1489.5)
(1.6, 1489.5)
(1.7, 1489.5)
(1.8, 1489.5)
(1.9, 1489.5)
(2.0, 1489.5)
(2.1, 1489.5)
(2.2, 1489.5)
(2.3, 1489.5)
(2.4, 1489.5)
(2.5, 1489.5)
};
\end{axis}

\end{tikzpicture}
  \caption{Boxplots show the distribution of the
    free-run simulation MSE over the validation window
    for models trained using series-parallel (SP) and parallel (P)
    methods under the circumstances specified in
    Figure~\ref{fig:pilotplant_val}. There are 100 realizations
    of the training in each boxplot,
    for each realization the weights $w_{i,j}^{(n)}$
    are drawn from a normal distribution with zero mean standard
    deviation ${\sigma=(N_{s (n-1)})^{-0.5}}$  and the bias terms
    $\gamma_i^{(n)}$ are initialized with zeros~\cite{lecun_efficient_2012}.
    For comparison purposes, the dashed horizontal
      line gives the performance of an ARX linear model ($n_y=1$ and $n_u=1$)
      trained and tested under the same conditions.}
  \label{fig:pilotplant_boxplots}
\end{figure}

The training of the neural network was performed using
normalized data. However, if  \textit{unscaled data} were used instead,
parallel training would yield models with
${\text{MSE} > 10000}$ over the validation window
while series-parallel training can still yield
solutions with a reasonable fit to the validation data. We understand this as another  indicator
of parallel training greater sensitivity to the initial parameter guess:
for unscaled data, the initial guess is far away from meaningful
solutions of the problem, and, while series-parallel training
converges to acceptable results, parallel training gets trapped in
``poor'' local solutions.

\subsection{Example 2: Investigating the noise effect}
The non-linear system:~\cite{chen_non-linear_1990}:
\begin{eqnarray}
  y^*[k] &=& (0.8-0.5e^{-y^*[k-1]^2})y^*[k-1]- \nonumber\\
    &&(0.3+0.9e^{-y^*[k-1]^2})y^*[k-2]+u[k-1]+ \nonumber\\
    &&0.2u[k-2]+0.1u[k-1]u[k-2] + v[k] \nonumber\\
y[k] &=& y^*[k] + w[k], \label{eq:ex_nonlinear_system}
\end{eqnarray}
was simulated and the generated dataset was used
to build neural network models. Figure~\ref{fig:simulation_val}
shows the validation results for models obtained  for a
training set generated with white Gaussian equation and output errors
$v$ and $w$. In this section,
we repeat this same experiment for
diverse random processes applied to $v$ and $w$ in order to investigate
how different noise configurations
affect parallel and series-parallel training.

\begin{figure}[h]
  \centering
  \begin{tikzpicture}
  \begin{axis}[
    yscale=0.5,
    y label style={rotate=-90, at={(0.03,1.2)}},
    x label style={at={(0.5,0.3)}},
    width=1\textwidth,
    xlabel=$k$,
    ylabel=$y$,
    xmin=0,
    xmax=100,
    ymin=-5,
    ymax=5,
    legend entries={$y$, SP, P},
    legend style={at={(1.0, 0.9)},
      anchor={north east}}]

    \addplot[
    color=black,
    no marks,
    line width=0.8pt]
    table [
    x expr=\thisrow{"time"},
    y expr=\thisrow{"output"},
    col sep=comma] {./data/validation_results_simulation.csv};\

    \addplot[
    color={rgb,1:red,0.00000000;green,0.50196078;blue,0.00000000},
    no marks,
    line width=1.2pt,
    style=dashed]
    table [
    x expr=\thisrow{"time"},
    y expr=\thisrow{"SP_simulation"},
    col sep=comma] {./data/validation_results_simulation.csv};\

    \addplot[
    color=blue,
    no marks,
    line width=1.2pt,
    style=dashed]
    table [
    x expr=\thisrow{"time"},
    y expr=\thisrow{"PPhi_simulation"},
    col sep=comma] {./data/validation_results_simulation.csv};\
  \end{axis}
\end{tikzpicture}
  \caption{Displays the first 100 samples of the free-run simulation in
    the validation window for models trained using series-parallel (SP) and parallel (P)
    methods. The mean square errors are ${\text{MSE}_{\text{SP}} = 0.39}$;
    ${\text{MSE}_{\text{P}} = 0.06}$.
    The models have ${n_y=n_u=2}$ and a single hidden layer with 10 nodes.
    The training set has $N= 1000$ samples and was generated
    with~(\ref{eq:ex_nonlinear_system}) for
    $v$ and $w$ white Gaussian noise with standard deviations $\sigma_v = 0.1$
    and $\sigma_w = 0.5$. 
    The validation window is generated without the noise effect.
    For both, the input $u$ is randomly generated with standard
    Gaussian distribution, each randomly generated value held for 5 samples.
    The training was 100 epochs long, which took 5.0 and 6.1 s
    for, respectively, series-parallel and parallel training.}
  \label{fig:simulation_val}
\end{figure}
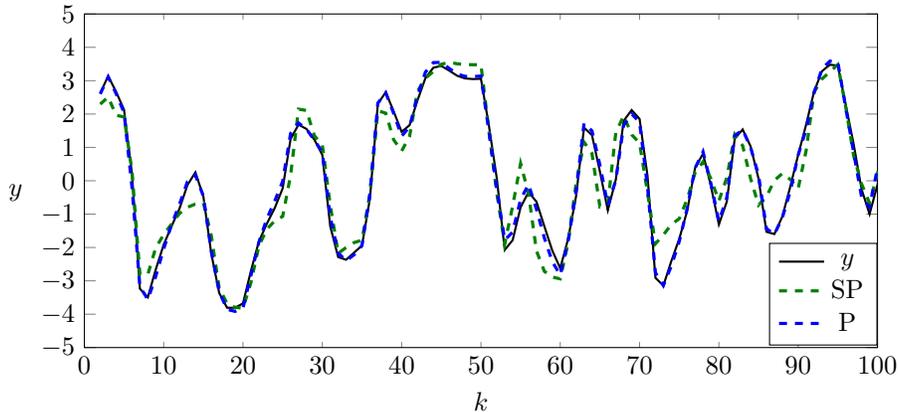

\subsubsection{White noise}

Let $v$ be white Gaussian noise with
standard deviation ${\sigma_v}$ and let $w$ be zero.
Figure~\ref{fig:mse_vs_noiselevel} (a) shows the free-run simulation error
on the validation window
using parallel and series-parallel training
for increasing values of ${\sigma_v}$.
Figure~\ref{fig:mse_vs_noiselevel} (b) shows the
complementary experiment, for which $v$ is zero and $w$
is white Gaussian noise with increasing
larger values of ${\sigma_w}$ being tried out.

In Section~\ref{sec:unifying-framework},
series-parallel training was derived considering
the presence of white equation error and, in this
situation, the numerical results
illustrate the model obtained using this training method presents the
best validation results (Figure~\ref{fig:mse_vs_noiselevel} (a)).
On the other hand, parallel training was derived
considering the presence of white output
error and is significantly better in
this alternative situation (Figure~\ref{fig:mse_vs_noiselevel} (b)).

\begin{figure}[H]
  \centering
  \subfloat[]{\begin{tikzpicture}[]
  \begin{axis}[height = {0.4\textwidth}, legend pos = {north west}, ylabel = {MSE}, xmin = {-0.06}, xmax = {2.06}, ymax = {2.0}, xlabel = {$\sigma_v$}, {unbounded coords=jump, xticklabel style={rotate = 0}, xtick = {0.0,0.4,0.8,1.2,1.6,2.0}, xticklabels = {0.0,0.4,0.8,1.2,1.6,2.0}, yticklabel style={rotate = 0}, ytick = {0.0,0.5,1.0,1.5,2.0}, yticklabels = {0.0,0.5,1.0,1.5,2.0},     xshift = 0.0mm,
      yshift = 0.0mm,
      axis background/.style={fill={rgb,1:red,1.00000000;green,1.00000000;blue,1.00000000}}
    }, ymin = {0}, width = {0.505\textwidth}]

    \addplot+ [name path = main_SP,
    color = {rgb,1:red,0.00000000;green,0.50196078;blue,0.00000000},
    draw opacity=1.0,
    line width=1,
    solid,mark = *,
    mark size = 2.0,
    mark options = {
      color = {rgb,1:red,0.00000000;green,0.00000000;blue,0.00000000}, draw opacity = 1.0,
      fill = {rgb,1:red,0.00000000;green,0.50196078;blue,0.00000000}, fill opacity = 1.0,
      line width = 1,
      rotate = 0,
      solid
    }]coordinates {
      (0.0, 0.016225494314657384)
      (0.1, 0.022359558615788776)
      (0.2, 0.04217774303987396)
      (0.3, 0.06391697028606087)
      (0.4, 0.11863792901988948)
      (0.5, 0.15204068910761212)
      (0.6, 0.17681857076728563)
      (0.7, 0.19770320766671423)
      (0.8, 0.2771097427464301)
      (0.9, 0.24991028276119812)
      (1.0, 0.3665896248733497)
      (1.1, 0.4226498416602078)
      (1.2, 0.40418474986115893)
      (1.3, 0.5373979159411937)
      (1.4, 0.5436503131527037)
      (1.5, 0.7292538070602619)
      (1.6, 0.7879262661755586)
      (1.7, 0.8079199580341745)
      (1.8, 0.8902141697814792)
      (1.9, 0.9784684005473158)
      (2.0, 1.2080233676768284)
    };
    \addlegendentry{SP}
    \addplot+ [name path = main_P,
    color = {rgb,1:red,0.00000000;green,0.00000000;blue,1.00000000},
    draw opacity=1.0,
    line width=1,
    solid,mark = diamond*,
    mark size = 2.0,
    mark options = {
      color = {rgb,1:red,0.00000000;green,0.00000000;blue,0.00000000}, draw opacity = 1.0,
      fill = {rgb,1:red,0.00000000;green,0.00000000;blue,1.00000000}, fill opacity = 1.0,
      line width = 1,
      rotate = 0,
      solid
    }]coordinates {
      (0.0, 0.022242322461811896)
      (0.1, 0.05184533215410653)
      (0.2, 0.11168081500942095)
      (0.3, 0.21194930336198797)
      (0.4, 0.25839614151248874)
      (0.5, 0.3030495995783567)
      (0.6, 0.39667025626288965)
      (0.7, 0.4024754450976278)
      (0.8, 0.47390856205551546)
      (0.9, 0.5021058147661347)
      (1.0, 0.8015745876702229)
      (1.1, 0.6744096277016616)
      (1.2, 0.5654553872686854)
      (1.3, 0.6680336347774838)
      (1.4, 0.709136976050873)
      (1.5, 0.8021634486775245)
      (1.6, 0.9647387960770698)
      (1.7, 1.008388819572539)
      (1.8, 1.038651900765326)
      (1.9, 1.293212947198259)
      (2.0, 1.4489064402185443)
    };
    \addlegendentry{P}
    \addplot+ [name path = upper_SP,
    color = {rgb,1:red,0.00000000;green,0.50196078;blue,0.00000000},
    draw opacity=0.5,
    line width=1,
    solid,mark = none,
    mark size = 2.0,
    mark options = {
      color = {rgb,1:red,0.00000000;green,0.00000000;blue,0.00000000}, draw opacity = 0.4,
      fill = {rgb,1:red,0.00000000;green,0.50196078;blue,0.00000000}, fill opacity = 0.4,
      line width = 1,
      rotate = 0,
      solid
    },forget plot]coordinates {
      (0.0, 0.020588903597037224)
      (0.1, 0.02515207437742513)
      (0.2, 0.05490064979563886)
      (0.3, 0.07633164715609685)
      (0.4, 0.14110340287547848)
      (0.5, 0.18381437331260583)
      (0.6, 0.2674349672084538)
      (0.7, 0.28351076054799884)
      (0.8, 0.33387023783408787)
      (0.9, 0.3494088180194121)
      (1.0, 0.46005644494184617)
      (1.1, 0.5574220281230544)
      (1.2, 0.5300692561781604)
      (1.3, 0.6375052558937103)
      (1.4, 0.8800277499023222)
      (1.5, 0.9814417873362941)
      (1.6, 0.9239358715193691)
      (1.7, 1.0376841518126532)
      (1.8, 1.0478888677101155)
      (1.9, 1.1882928281098348)
      (2.0, 1.2392351942627353)
    };
    \addplot+ [name path = bottom_SP,
    color = {rgb,1:red,0.00000000;green,0.50196078;blue,0.00000000},
    draw opacity=0.5,
    line width=1,
    solid,mark = none,
    mark size = 2.0,
    mark options = {
      color = {rgb,1:red,0.00000000;green,0.00000000;blue,0.00000000}, draw opacity = 0.3,
      fill = {rgb,1:red,0.00000000;green,0.50196078;blue,0.00000000}, fill opacity = 0.3,
      line width = 1,
      rotate = 0,
      solid
    },forget plot]coordinates{
      (0.0, 0.014372810867415799)
      (0.1, 0.01672088185973898)
      (0.2, 0.03097352371380085)
      (0.3, 0.04407328489366489)
      (0.4, 0.081880563374839)
      (0.5, 0.11289390678077386)
      (0.6, 0.1346960680680404)
      (0.7, 0.1925797244937627)
      (0.8, 0.22889337858169082)
      (0.9, 0.19547872419289358)
      (1.0, 0.2423725647143735)
      (1.1, 0.32919985314066685)
      (1.2, 0.4018177395416466)
      (1.3, 0.4961504096592875)
      (1.4, 0.42354758763683786)
      (1.5, 0.6016725626663578)
      (1.6, 0.6323849033283694)
      (1.7, 0.7154847421460885)
      (1.8, 0.7958548194507806)
      (1.9, 0.8445727142630833)
      (2.0, 1.0443458075818788)
    };
    \addplot+ [name path = upper_P,
    color = {rgb,1:red,0.00000000;green,0.00000000;blue,1.00000000},
    draw opacity=0.5,
    line width=1,
    solid,mark = none,
    mark size = 2.0,
    mark options = {
      color = {rgb,1:red,0.00000000;green,0.00000000;blue,0.00000000}, draw opacity = 0.4,
      fill = {rgb,1:red,0.00000000;green,0.00000000;blue,1.00000000}, fill opacity = 0.4,
      line width = 1,
      rotate = 0,
      solid
    },forget plot]coordinates {
      (0.0, 0.026227605073744917)
      (0.1, 0.06823999330627924)
      (0.2, 0.13198281063618347)
      (0.3, 0.22561458135040993)
      (0.4, 0.276431408936666)
      (0.5, 0.36217254673099003)
      (0.6, 0.5233874225983536)
      (0.7, 0.536354096878373)
      (0.8, 0.49637261116239717)
      (0.9, 0.7238415626689398)
      (1.0, 0.9744399954799149)
      (1.1, 1.0349721197317465)
      (1.2, 0.6324813374363608)
      (1.3, 0.9394511348299988)
      (1.4, 1.2690777399492765)
      (1.5, 1.0699277968703185)
      (1.6, 1.186987079118639)
      (1.7, 1.2789825332903475)
      (1.8, 1.4793546950631482)
      (1.9, 1.6931293673829282)
      (2.0, 1.6262202776870787)
    };
    \addplot+ [name path = bottom_P,
    color = {rgb,1:red,0.00000000;green,0.00000000;blue,1.00000000},
    draw opacity=0.5,
    line width=1,
    solid,mark = none,
    mark size = 2.0,
    mark options = {
      color = {rgb,1:red,0.00000000;green,0.00000000;blue,0.00000000}, draw opacity = 0.3,
      fill = {rgb,1:red,0.00000000;green,0.00000000;blue,1.00000000}, fill opacity = 0.3,
      line width = 1,
      rotate = 0,
      solid
    },forget plot]coordinates {
      (0.0, 0.014501240891213373)
      (0.1, 0.03688843636371717)
      (0.2, 0.08634788275099749)
      (0.3, 0.16431629863414557)
      (0.4, 0.21513460593986625)
      (0.5, 0.26230107936903435)
      (0.6, 0.3662307017186598)
      (0.7, 0.3581007960345759)
      (0.8, 0.39456518430032206)
      (0.9, 0.41352894972495324)
      (1.0, 0.6294477126619195)
      (1.1, 0.5274292757227912)
      (1.2, 0.5462632820735652)
      (1.3, 0.6170043733480742)
      (1.4, 0.6176571350293322)
      (1.5, 0.6638450441143596)
      (1.6, 0.722001078731206)
      (1.7, 0.7869302478200558)
      (1.8, 0.7858958004155939)
      (1.9, 0.8824121829112225)
      (2.0, 1.187586354345072)
    };

    \addplot[color = {rgb,1:red,0.00000000;green,0.50196078;blue,0.00000000},
    fill opacity=0.3]  fill between[ 
    of = main_SP and upper_SP
    ];

    \addplot[color = {rgb,1:red,0.00000000;green,0.50196078;blue,0.00000000},
    fill opacity=0.3]  fill between[ 
    of = main_SP and bottom_SP
    ];

    \addplot[color = {rgb,1:red,0.00000000;green,0.00000000;blue,1.00000000},
    fill opacity=0.3]  fill between[ 
    of = main_P and upper_P
    ];

    \addplot[color = {rgb,1:red,0.00000000;green,0.00000000;blue,1.00000000},
    fill opacity=0.3]  fill between[ 
    of = main_P and bottom_P
    ];
  \end{axis}

\end{tikzpicture}}
  \subfloat[]{\begin{tikzpicture}[]
  \begin{axis}[height = {0.4\textwidth}, ylabel = {}, xmin = {-0.06}, xmax = {2.06}, ymax = {1.0}, xlabel = {$\sigma_w$}, {unbounded coords=jump, xticklabel style={rotate = 0}, xtick = {0.0,0.4,0.8,1.2,1.6,2.0}, xticklabels = {0.0,0.4,0.8,1.2,1.6,2.0}, yticklabel style={rotate = 0}, ytick = {0.0,0.2,0.4,0.6,0.8,1.0}, yticklabels = {0.0,0.2,0.4,0.6,0.8,1.0},     xshift = 0.0mm,
      yshift = 0.0mm,
      axis background/.style={fill={rgb,1:red,1.00000000;green,1.00000000;blue,1.00000000}}
    }, ymin = {0}, width = {0.495\textwidth}]
    \addplot+ [name path = main_SP,
    color = {rgb,1:red,0.00000000;green,0.50196078;blue,0.00000000},
    draw opacity=1.0,
    line width=1,
    solid,mark = *,
    mark size = 2.0,
    mark options = {
      color = {rgb,1:red,0.00000000;green,0.00000000;blue,0.00000000}, draw opacity = 1.0,
      fill = {rgb,1:red,0.00000000;green,0.50196078;blue,0.00000000}, fill opacity = 1.0,
      line width = 1,
      rotate = 0,
      solid
    },forget plot]coordinates {
      (0.0, 0.016225494314657384)
      (0.1, 0.029709361231770662)
      (0.2, 0.10754348046255127)
      (0.3, 0.20152950187800106)
      (0.4, 0.2693290899024695)
      (0.5, 0.37631274049606184)
      (0.6, 0.4166231625236998)
      (0.7, 0.4339324653067199)
      (0.8, 0.4936904493024121)
      (0.9, 0.5415792164986866)
      (1.0, 0.5830567289016189)
      (1.1, 0.638533675244709)
      (1.2, 0.6614567959350404)
      (1.3, 0.6614332831276419)
      (1.4, 0.6864182807864496)
      (1.5, 0.6875997234684363)
      (1.6, 0.6907649619117024)
      (1.7, 0.764088533311996)
      (1.8, 0.8150960037431755)
      (1.9, 0.760261180276585)
      (2.0, 0.7774982192873281)
    };
    \addplot+ [name path = main_P,
    color = {rgb,1:red,0.00000000;green,0.00000000;blue,1.00000000},
    draw opacity=1.0,
    line width=1,
    solid,mark = diamond*,
    mark size = 2.0,
    mark options = {
      color = {rgb,1:red,0.00000000;green,0.00000000;blue,0.00000000}, draw opacity = 1.0,
      fill = {rgb,1:red,0.00000000;green,0.00000000;blue,1.00000000}, fill opacity = 1.0,
      line width = 1,
      rotate = 0,
      solid
    },forget plot]coordinates {
      (0.0, 0.022242322461811896)
      (0.1, 0.02170898488450872)
      (0.2, 0.01798064138808412)
      (0.3, 0.02126765832399221)
      (0.4, 0.026832544123704588)
      (0.5, 0.042830735422189486)
      (0.6, 0.04564258976778447)
      (0.7, 0.0610614996854416)
      (0.8, 0.07837297817685851)
      (0.9, 0.09260275882319186)
      (1.0, 0.12435413943863391)
      (1.1, 0.1573055193441094)
      (1.2, 0.12634750582347964)
      (1.3, 0.13693987953458364)
      (1.4, 0.16034163068213636)
      (1.5, 0.17462214051406416)
      (1.6, 0.20992536959262384)
      (1.7, 0.24044759073399719)
      (1.8, 0.29046584368816397)
      (1.9, 0.29894717700344214)
      (2.0, 0.34867865080773264)
    };
    \addplot+ [name path = upper_SP,
    color = {rgb,1:red,0.00000000;green,0.50196078;blue,0.00000000},
    draw opacity=0.5,
    line width=1,
    solid,mark = none,
    mark size = 2.0,
    mark options = {
      color = {rgb,1:red,0.00000000;green,0.00000000;blue,0.00000000}, draw opacity = 0.4,
      fill = {rgb,1:red,0.00000000;green,0.50196078;blue,0.00000000}, fill opacity = 0.4,
      line width = 1,
      rotate = 0,
      solid
    },forget plot]coordinates {
      (0.0, 0.020588903597037224)
      (0.1, 0.035468128871490506)
      (0.2, 0.11630597753090771)
      (0.3, 0.22716166020222722)
      (0.4, 0.3053505631654298)
      (0.5, 0.40898784242681036)
      (0.6, 0.4510754363650529)
      (0.7, 0.47356533954561536)
      (0.8, 0.537448958912624)
      (0.9, 0.5884427259502932)
      (1.0, 0.6121511101799948)
      (1.1, 0.6612396364517916)
      (1.2, 0.7053063798364586)
      (1.3, 0.7107459221736803)
      (1.4, 0.7283972648436977)
      (1.5, 0.7389406979510518)
      (1.6, 0.7741495091328416)
      (1.7, 0.7696014719397164)
      (1.8, 0.8416387895171589)
      (1.9, 0.8567336501769207)
      (2.0, 0.8660043225051968)
    };
    \addplot+ [name path = bottom_SP,
    color = {rgb,1:red,0.00000000;green,0.50196078;blue,0.00000000},
    draw opacity=0.5,
    line width=1,
    solid,mark = none,
    mark size = 2.0,
    mark options = {
      color = {rgb,1:red,0.00000000;green,0.00000000;blue,0.00000000}, draw opacity = 0.3,
      fill = {rgb,1:red,0.00000000;green,0.50196078;blue,0.00000000}, fill opacity = 0.3,
      line width = 1,
      rotate = 0,
      solid
    },forget plot]coordinates {
      (0.0, 0.014372810867415799)
      (0.1, 0.028061249948976998)
      (0.2, 0.10017617768403808)
      (0.3, 0.17324870881580295)
      (0.4, 0.2501773071517106)
      (0.5, 0.30944137070063366)
      (0.6, 0.34953421860428924)
      (0.7, 0.4200895638654015)
      (0.8, 0.4620301699609242)
      (0.9, 0.4967861848841082)
      (1.0, 0.5448643322150257)
      (1.1, 0.5692427167308022)
      (1.2, 0.6234036468844658)
      (1.3, 0.5900879483935615)
      (1.4, 0.630000600369462)
      (1.5, 0.6724224183552925)
      (1.6, 0.639481126425478)
      (1.7, 0.7485736625206915)
      (1.8, 0.6893865376829166)
      (1.9, 0.7027732731497819)
      (2.0, 0.7000263954863697)
    };
    \addplot+ [name path = upper_P,
    color = {rgb,1:red,0.00000000;green,0.00000000;blue,1.00000000},
    draw opacity=0.5,
    line width=1,
    solid,mark = none,
    mark size = 2.0,
    mark options = {
      color = {rgb,1:red,0.00000000;green,0.00000000;blue,0.00000000}, draw opacity = 0.4,
      fill = {rgb,1:red,0.00000000;green,0.00000000;blue,1.00000000}, fill opacity = 0.4,
      line width = 1,
      rotate = 0,
      solid
    },forget plot]coordinates {
      (0.0, 0.026227605073744917)
      (0.1, 0.03837845866690375)
      (0.2, 0.022396253998839856)
      (0.3, 0.03756821451793006)
      (0.4, 0.03905651599221972)
      (0.5, 0.05054507282703609)
      (0.6, 0.059513646162343364)
      (0.7, 0.0671299054732527)
      (0.8, 0.09435207917421579)
      (0.9, 0.10939264188335986)
      (1.0, 0.15166478817702947)
      (1.1, 0.20331627289645698)
      (1.2, 0.16881043003688692)
      (1.3, 0.1629503019417265)
      (1.4, 0.1895148717052683)
      (1.5, 0.19651760559200987)
      (1.6, 0.25496054215469166)
      (1.7, 0.25708878005883923)
      (1.8, 0.31757410277545794)
      (1.9, 0.4285951389725843)
      (2.0, 0.46196717613865856)
    };
    \addplot+ [name path = bottom_P,
    color = {rgb,1:red,0.00000000;green,0.00000000;blue,1.00000000},
    draw opacity=0.5,
    line width=1,
    solid,mark = none,
    mark size = 2.0,
    mark options = {
      color = {rgb,1:red,0.00000000;green,0.00000000;blue,0.00000000}, draw opacity = 0.3,
      fill = {rgb,1:red,0.00000000;green,0.00000000;blue,1.00000000}, fill opacity = 0.3,
      line width = 1,
      rotate = 0,
      solid
    },forget plot]coordinates {
      (0.0, 0.014501240891213373)
      (0.1, 0.015938621284997383)
      (0.2, 0.014122935463359758)
      (0.3, 0.018818507785536343)
      (0.4, 0.023953048678544944)
      (0.5, 0.03612347896572172)
      (0.6, 0.03973849323031346)
      (0.7, 0.05212715881745475)
      (0.8, 0.060171580894997495)
      (0.9, 0.06836056816337185)
      (1.0, 0.10274950264957874)
      (1.1, 0.09883535426051077)
      (1.2, 0.10890145696920452)
      (1.3, 0.11510366554432408)
      (1.4, 0.1280018227468977)
      (1.5, 0.1634051380359211)
      (1.6, 0.1599435597570715)
      (1.7, 0.20604961450782056)
      (1.8, 0.2654658846342724)
      (1.9, 0.2568610303300261)
      (2.0, 0.3161896651286735)
    };

    \addplot[color = {rgb,1:red,0.00000000;green,0.50196078;blue,0.00000000},
        fill opacity=0.3]  fill between[ 
    of = main_SP and upper_SP
    ];

\addplot[color = {rgb,1:red,0.00000000;green,0.50196078;blue,0.00000000},
        fill opacity=0.3]  fill between[ 
    of = main_SP and bottom_SP
    ];

\addplot[color = {rgb,1:red,0.00000000;green,0.00000000;blue,1.00000000},
        fill opacity=0.3]  fill between[ 
    of = main_P and upper_P
    ];

\addplot[color = {rgb,1:red,0.00000000;green,0.00000000;blue,1.00000000},
        fill opacity=0.3]  fill between[ 
    of = main_P and bottom_P
  ];
  \end{axis}

\end{tikzpicture}}
  \caption{Free-run simulation MSE on
    the validation window  \textit{vs} noise levels
    for series-parallel and parallel training.
    The main line indicates the median and the shaded region
    indicates interquartile range. These statistics were computed from 12 realizations.
    In (a) $v$ is a Gaussian white process and $w=0$; and, in (b) $w$ is
    a Gaussian white process and $v=0$. }
  \label{fig:mse_vs_noiselevel}
\end{figure}
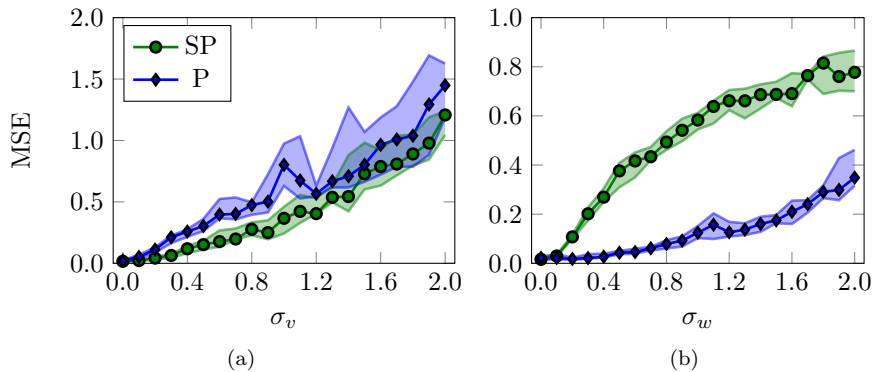

\subsubsection{Colored noise}

Consider $w=0$ and $v$ a white Gaussian noise filtered by a low pass
filter with cutoff frequency $\omega_c$.
Figure~\ref{fig:mse_vs_noiselevel_colored} shows the free-run simulation error
in the validation window for both parallel and series-parallel training for a
sequence of values of $\omega_c$ and different noise intensities.
The result indicates parallel training provides the best results unless
the equation error has a very large bandwidth.

More extensive tests are
summarized in Table~\ref{tab:filtered_noise}, which shows
the validation errors for a training set with colored Gaussian errors
in different frequency bands. Again, except
for white or large bandwidth equation error, parallel training seems
to provide the smallest validation errors.

\begin{figure}[t]
  \centering
  \subfloat[$\omega_c = 0.8$]{\begin{tikzpicture}[]
  \begin{axis}[height = {0.4\textwidth}, legend pos = {north west}, ylabel = {MSE}, xmin = {0}, xmax = {2}, ymax = {8.0}, xlabel = {}, {unbounded coords=jump, xticklabel style={rotate = 0}, xtick = {}, xticklabels = {}, yticklabel style={rotate = 0}, ytick = {0.0,2.0,4.0,6.0,8.0}, yticklabels = {0,2,4,6,8},     xshift = 0.0mm,
      yshift = 0.0mm,
      axis background/.style={fill={rgb,1:red,1.00000000;green,1.00000000;blue,1.00000000}}
    }, ymin = {0}, width = {0.505\textwidth}]
    \addplot+ [name path = main_SP,
    color = {rgb,1:red,0.00000000;green,0.50196078;blue,0.00000000},
    draw opacity=1.0,
    line width=1,
    solid,mark = *,
    mark size = 2.0,
    mark options = {
      color = {rgb,1:red,0.00000000;green,0.00000000;blue,0.00000000}, draw opacity = 1.0,
      fill = {rgb,1:red,0.00000000;green,0.50196078;blue,0.00000000}, fill opacity = 1.0,
      line width = 1,
      rotate = 0,
      solid
    }]coordinates {
      (0.0, 0.016225494314657384)
      (0.1, 0.02061692393927777)
      (0.2, 0.0468200702583625)
      (0.3, 0.07446969995072009)
      (0.4, 0.11992149948852426)
      (0.5, 0.19264249091885022)
      (0.6, 0.27020573209778387)
      (0.7, 0.267484512550901)
      (0.8, 0.3628348270145753)
      (0.9, 0.46204636257371845)
      (1.0, 0.5142454999038238)
      (1.1, 0.6140593522971896)
      (1.2, 0.6936798022813697)
      (1.3, 0.8898935286451813)
      (1.4, 0.8441200318913031)
      (1.5, 0.8934640971541346)
      (1.6, 1.1002911539582318)
      (1.7, 1.239044583296605)
      (1.8, 1.3038882330758697)
      (1.9, 1.2482213900825687)
      (2.0, 1.299911320008993)
    };
    \addlegendentry{SP}
    \addplot+ [name path = main_P,
    color = {rgb,1:red,0.00000000;green,0.00000000;blue,1.00000000},
    draw opacity=1.0,
    line width=1,
    solid,mark = diamond*,
    mark size = 2.0,
    mark options = {
      color = {rgb,1:red,0.00000000;green,0.00000000;blue,0.00000000}, draw opacity = 1.0,
      fill = {rgb,1:red,0.00000000;green,0.00000000;blue,1.00000000}, fill opacity = 1.0,
      line width = 1,
      rotate = 0,
      solid
    }]coordinates {
      (0.0, 0.022242322461811896)
      (0.1, 0.0745592112561846)
      (0.2, 0.13909467126180636)
      (0.3, 0.24099293540119918)
      (0.4, 0.26498036120208585)
      (0.5, 0.39982207199474484)
      (0.6, 0.3818243931519125)
      (0.7, 0.458865381088831)
      (0.8, 0.6987656253744482)
      (0.9, 0.46115123236727457)
      (1.0, 0.5804160575357364)
      (1.1, 0.585488703277287)
      (1.2, 0.8094517491463294)
      (1.3, 0.8318392680644171)
      (1.4, 0.7494662730202362)
      (1.5, 1.3035532292443373)
      (1.6, 1.313550788974891)
      (1.7, 1.0281082170819102)
      (1.8, 1.2839882415279371)
      (1.9, 1.1925934515201475)
      (2.0, 1.2585194959847672)
    };
    \addlegendentry{P}
    \addplot+ [name path = upper_SP,
    color = {rgb,1:red,0.00000000;green,0.50196078;blue,0.00000000},
    draw opacity=0.5,
    line width=1,
    solid,mark = none,
    mark size = 2.0,
    mark options = {
      color = {rgb,1:red,0.00000000;green,0.00000000;blue,0.00000000}, draw opacity = 0.4,
      fill = {rgb,1:red,0.00000000;green,0.50196078;blue,0.00000000}, fill opacity = 0.4,
      line width = 1,
      rotate = 0,
      solid
    },forget plot]coordinates {
      (0.0, 0.020588903597037224)
      (0.1, 0.02901687414625695)
      (0.2, 0.055831491639685184)
      (0.3, 0.08200711286650406)
      (0.4, 0.14329400325889513)
      (0.5, 0.2193398759913224)
      (0.6, 0.34306582546264885)
      (0.7, 0.37542908265426994)
      (0.8, 0.4650071724005445)
      (0.9, 0.5291164230724994)
      (1.0, 0.6567302632490916)
      (1.1, 0.7548228938661767)
      (1.2, 0.7547812906499952)
      (1.3, 0.9648234158472972)
      (1.4, 0.9998667943458155)
      (1.5, 1.1834538288363055)
      (1.6, 1.1628668968178577)
      (1.7, 1.3776318914929104)
      (1.8, 1.579424037001624)
      (1.9, 1.375682968769867)
      (2.0, 1.3776928158282875)
    };
    \addplot+ [name path = bottom_SP,
    color = {rgb,1:red,0.00000000;green,0.50196078;blue,0.00000000},
    draw opacity=0.5,
    line width=1,
    solid,mark = none,
    mark size = 2.0,
    mark options = {
      color = {rgb,1:red,0.00000000;green,0.00000000;blue,0.00000000}, draw opacity = 0.3,
      fill = {rgb,1:red,0.00000000;green,0.50196078;blue,0.00000000}, fill opacity = 0.3,
      line width = 1,
      rotate = 0,
      solid
    },forget plot]coordinates {
      (0.0, 0.014372810867415799)
      (0.1, 0.01657667964189209)
      (0.2, 0.03719446444560084)
      (0.3, 0.06337328278398108)
      (0.4, 0.09643699362162073)
      (0.5, 0.15472894879716148)
      (0.6, 0.21968375082426703)
      (0.7, 0.24757150435915443)
      (0.8, 0.30346465745200735)
      (0.9, 0.3789888109010463)
      (1.0, 0.4081339286688241)
      (1.1, 0.5387119309323977)
      (1.2, 0.6565485180094258)
      (1.3, 0.7658875947792299)
      (1.4, 0.7889578767370683)
      (1.5, 0.7586503271778173)
      (1.6, 0.961328245489819)
      (1.7, 0.8131152940051692)
      (1.8, 1.0755817024117484)
      (1.9, 1.2231464597637258)
      (2.0, 1.1041141661845846)
    };
    \addplot+ [name path = upper_P,color = {rgb,1:red,0.00000000;green,0.00000000;blue,1.00000000},
    draw opacity=0.5,
    line width=1,
    solid,mark = none,
    mark size = 2.0,
    mark options = {
      color = {rgb,1:red,0.00000000;green,0.00000000;blue,0.00000000}, draw opacity = 0.4,
      fill = {rgb,1:red,0.00000000;green,0.00000000;blue,1.00000000}, fill opacity = 0.4,
      line width = 1,
      rotate = 0,
      solid
    },forget plot]coordinates {
      (0.0, 0.026227605073744917)
      (0.1, 0.08772051116963504)
      (0.2, 0.16446498259425257)
      (0.3, 0.2844378357866019)
      (0.4, 0.28632285290049964)
      (0.5, 0.4202792346928502)
      (0.6, 0.45735536719920744)
      (0.7, 0.5320815446456355)
      (0.8, 1.0619099935814833)
      (0.9, 0.5593187152143754)
      (1.0, 0.6431371351543458)
      (1.1, 0.7469268272444552)
      (1.2, 1.2206405104194333)
      (1.3, 1.5524456441676169)
      (1.4, 1.2376717338388434)
      (1.5, 2.2532557578593795)
      (1.6, 1.4371692590410314)
      (1.7, 1.2192834596086255)
      (1.8, 1.5888212354942395)
      (1.9, 1.5865830586559042)
      (2.0, 1.9925089626809136)
    };
    \addplot+ [name path = bottom_P,
    color = {rgb,1:red,0.00000000;green,0.00000000;blue,1.00000000},
    draw opacity=0.5,
    line width=1,
    solid,mark = none,
    mark size = 2.0,
    mark options = {
      color = {rgb,1:red,0.00000000;green,0.00000000;blue,0.00000000}, draw opacity = 0.3,
      fill = {rgb,1:red,0.00000000;green,0.00000000;blue,1.00000000}, fill opacity = 0.3,
      line width = 1,
      rotate = 0,
      solid
    },forget plot]coordinates {
      (0.0, 0.014501240891213373)
      (0.1, 0.05266734328852329)
      (0.2, 0.104799220806525)
      (0.3, 0.1930024846910563)
      (0.4, 0.24340843768033052)
      (0.5, 0.32917308148226015)
      (0.6, 0.36338979515058045)
      (0.7, 0.3872908273050767)
      (0.8, 0.5182305371969754)
      (0.9, 0.4445452062764892)
      (1.0, 0.49578454240665104)
      (1.1, 0.49099894443173353)
      (1.2, 0.6496829270107637)
      (1.3, 0.6253786057901571)
      (1.4, 0.6410524123335483)
      (1.5, 0.9380928070823895)
      (1.6, 1.0035717716438568)
      (1.7, 0.8397195842566494)
      (1.8, 0.8543044452752767)
      (1.9, 0.9820823871043698)
      (2.0, 0.9992170415806314)
    };
    \addplot[color = {rgb,1:red,0.00000000;green,0.50196078;blue,0.00000000},
    fill opacity=0.3]  fill between[ 
    of = main_SP and upper_SP
    ];

    \addplot[color = {rgb,1:red,0.00000000;green,0.50196078;blue,0.00000000},
    fill opacity=0.3]  fill between[ 
    of = main_SP and bottom_SP
    ];

    \addplot[color = {rgb,1:red,0.00000000;green,0.00000000;blue,1.00000000},
    fill opacity=0.3]  fill between[ 
    of = main_P and upper_P
    ];

    \addplot[color = {rgb,1:red,0.00000000;green,0.00000000;blue,1.00000000},
    fill opacity=0.3]  fill between[ 
    of = main_P and bottom_P
    ];
  \end{axis}

\end{tikzpicture}}
  \subfloat[$\omega_c = 0.6$]{\begin{tikzpicture}[]
  \begin{axis}[height = {0.4\textwidth}, ylabel = {}, xmin = {0}, xmax = {2}, ymax = {8.0}, xlabel = {}, {unbounded coords=jump, xticklabel style={rotate = 0}, xtick = {}, xticklabels = {}, yticklabel style={rotate = 0}, ytick = {}, yticklabels = {},     xshift = 0.0mm,
      yshift = 0.0mm,
      axis background/.style={fill={rgb,1:red,1.00000000;green,1.00000000;blue,1.00000000}}
    }, ymin = {0}, width = {0.495\textwidth}]
    \addplot+ [name path = main_SP,
    color = {rgb,1:red,0.00000000;green,0.50196078;blue,0.00000000},
    draw opacity=1.0,
    line width=1,
    solid,mark = *,
    mark size = 2.0,
    mark options = {
      color = {rgb,1:red,0.00000000;green,0.00000000;blue,0.00000000}, draw opacity = 1.0,
      fill = {rgb,1:red,0.00000000;green,0.50196078;blue,0.00000000}, fill opacity = 1.0,
      line width = 1,
      rotate = 0,
      solid
    },forget plot]coordinates {
      (0.0, 0.016225494314657384)
      (0.1, 0.0243244730605625)
      (0.2, 0.07348702088837106)
      (0.3, 0.16837640757757316)
      (0.4, 0.2621541879129581)
      (0.5, 0.36931501541235023)
      (0.6, 0.4822393517560576)
      (0.7, 0.5961653295980522)
      (0.8, 0.6660243076899739)
      (0.9, 0.8537123677807872)
      (1.0, 0.9135955173723805)
      (1.1, 1.111748825418049)
      (1.2, 1.1862559523815497)
      (1.3, 1.3809610355296882)
      (1.4, 1.6309449833499134)
      (1.5, 1.6518403137207138)
      (1.6, 1.6166617903873974)
      (1.7, 2.052266321512062)
      (1.8, 1.6486757905641234)
      (1.9, 1.9094937342031544)
      (2.0, 2.677012768089722)
    };
    \addplot+ [name path = main_P,
    color = {rgb,1:red,0.00000000;green,0.00000000;blue,1.00000000},
    draw opacity=1.0,
    line width=1,
    solid,mark = diamond*,
    mark size = 2.0,
    mark options = {
      color = {rgb,1:red,0.00000000;green,0.00000000;blue,0.00000000}, draw opacity = 1.0,
      fill = {rgb,1:red,0.00000000;green,0.00000000;blue,1.00000000}, fill opacity = 1.0,
      line width = 1,
      rotate = 0,
      solid
    },forget plot]coordinates {
      (0.0, 0.022242322461811896)
      (0.1, 0.08247549946344016)
      (0.2, 0.14922332331042815)
      (0.3, 0.2339861872403477)
      (0.4, 0.34507795804447455)
      (0.5, 0.40606386045432796)
      (0.6, 0.38695080748823873)
      (0.7, 0.6055438486065226)
      (0.8, 0.5628974777559692)
      (0.9, 0.6700793543153732)
      (1.0, 0.682928329349743)
      (1.1, 0.7954646896133863)
      (1.2, 1.0066206043648853)
      (1.3, 0.7801334989416147)
      (1.4, 1.0201613634964037)
      (1.5, 0.9454344354647767)
      (1.6, 1.1133082615102659)
      (1.7, 1.4043742637618628)
      (1.8, 1.3224240203318043)
      (1.9, 1.1681012636564545)
      (2.0, 1.2128811745029027)
    };
    \addplot+ [name path = upper_SP,
    color = {rgb,1:red,0.00000000;green,0.50196078;blue,0.00000000},
    draw opacity=0.5,
    line width=1,
    solid,mark = none,
    mark size = 2.0,
    mark options = {
      color = {rgb,1:red,0.00000000;green,0.00000000;blue,0.00000000}, draw opacity = 0.4,
      fill = {rgb,1:red,0.00000000;green,0.50196078;blue,0.00000000}, fill opacity = 0.4,
      line width = 1,
      rotate = 0,
      solid
    },forget plot]coordinates {
      (0.0, 0.020588903597037224)
      (0.1, 0.03382182073878512)
      (0.2, 0.08050414048366167)
      (0.3, 0.1849019103584818)
      (0.4, 0.35053839725360614)
      (0.5, 0.43937281333152045)
      (0.6, 0.5429407942685205)
      (0.7, 0.6390064165218925)
      (0.8, 0.708839688776251)
      (0.9, 0.886079256756371)
      (1.0, 1.0989249275478914)
      (1.1, 1.2384901786024878)
      (1.2, 1.5097980312294395)
      (1.3, 1.557930685076246)
      (1.4, 1.8165512517176325)
      (1.5, 1.765603619807088)
      (1.6, 2.198814550780948)
      (1.7, 2.248877077381038)
      (1.8, 2.295780403178767)
      (1.9, 2.280801165708553)
      (2.0, 2.9210349056110503)
    };
    \addplot+ [name path = bottom_SP,
    color = {rgb,1:red,0.00000000;green,0.50196078;blue,0.00000000},
    draw opacity=0.5,
    line width=1,
    solid,mark = none,
    mark size = 2.0,
    mark options = {
      color = {rgb,1:red,0.00000000;green,0.00000000;blue,0.00000000}, draw opacity = 0.3,
      fill = {rgb,1:red,0.00000000;green,0.50196078;blue,0.00000000}, fill opacity = 0.3,
      line width = 1,
      rotate = 0,
      solid
    },forget plot]coordinates {
      (0.0, 0.014372810867415799)
      (0.1, 0.021407727641958772)
      (0.2, 0.048812140804100596)
      (0.3, 0.11397224333228045)
      (0.4, 0.15805479798416475)
      (0.5, 0.30602702061601444)
      (0.6, 0.43340114853610584)
      (0.7, 0.5386418862845174)
      (0.8, 0.6395943452623528)
      (0.9, 0.7417847388033603)
      (1.0, 0.8138851097179874)
      (1.1, 0.9659156062149619)
      (1.2, 1.076426812803955)
      (1.3, 1.1428047878163887)
      (1.4, 1.1957604037394123)
      (1.5, 1.2635871471448838)
      (1.6, 1.3294420960831994)
      (1.7, 1.314861221720634)
      (1.8, 1.3319275776512434)
      (1.9, 1.6337467442661118)
      (2.0, 2.1808293855433867)
    };
    \addplot+ [name path = upper_P,
    color = {rgb,1:red,0.00000000;green,0.00000000;blue,1.00000000},
    draw opacity=0.5,
    line width=1,
    solid,mark = none,
    mark size = 2.0,
    mark options = {
      color = {rgb,1:red,0.00000000;green,0.00000000;blue,0.00000000}, draw opacity = 0.4,
      fill = {rgb,1:red,0.00000000;green,0.00000000;blue,1.00000000}, fill opacity = 0.4,
      line width = 1,
      rotate = 0,
      solid
    },forget plot]coordinates {
      (0.0, 0.026227605073744917)
      (0.1, 0.11641637846079521)
      (0.2, 0.16498381431477194)
      (0.3, 0.25902318598219143)
      (0.4, 0.44027315319430965)
      (0.5, 0.4984662194436375)
      (0.6, 0.5259033412937743)
      (0.7, 0.8110042649471558)
      (0.8, 0.7798459213749339)
      (0.9, 1.045131730691404)
      (1.0, 1.0827624837093919)
      (1.1, 0.8652319480498547)
      (1.2, 1.2989530411064802)
      (1.3, 0.9769048231741023)
      (1.4, 1.7497942212048896)
      (1.5, 1.3729029298675217)
      (1.6, 1.7454745112414125)
      (1.7, 1.629499124423434)
      (1.8, 1.916891224356976)
      (1.9, 2.2174314400261403)
      (2.0, 1.4521991526269526)
    };
    \addplot+ [name path = bottom_P,
    color = {rgb,1:red,0.00000000;green,0.00000000;blue,1.00000000},
    draw opacity=0.5,
    line width=1,
    solid,mark = none,
    mark size = 2.0,
    mark options = {
      color = {rgb,1:red,0.00000000;green,0.00000000;blue,0.00000000}, draw opacity = 0.3,
      fill = {rgb,1:red,0.00000000;green,0.00000000;blue,1.00000000}, fill opacity = 0.3,
      line width = 1,
      rotate = 0,
      solid
    },forget plot]coordinates {
      (0.0, 0.014501240891213373)
      (0.1, 0.06997372940038009)
      (0.2, 0.14213375476306866)
      (0.3, 0.19915188042781207)
      (0.4, 0.28480313246072836)
      (0.5, 0.3602785051086736)
      (0.6, 0.36310479839424914)
      (0.7, 0.5212974609272512)
      (0.8, 0.515514835777548)
      (0.9, 0.4711849785308858)
      (1.0, 0.5574979080931022)
      (1.1, 0.5605878176651885)
      (1.2, 0.8719959146393754)
      (1.3, 0.6688475842624257)
      (1.4, 0.7103367318434671)
      (1.5, 0.7302839082133018)
      (1.6, 1.0205463700553365)
      (1.7, 1.0716472054617494)
      (1.8, 1.1396525038342422)
      (1.9, 1.0245316825016868)
      (2.0, 0.8992220104178357)
    };
    \addplot[color = {rgb,1:red,0.00000000;green,0.50196078;blue,0.00000000},
    fill opacity=0.3]  fill between[ 
    of = main_SP and upper_SP
    ];

    \addplot[color = {rgb,1:red,0.00000000;green,0.50196078;blue,0.00000000},
    fill opacity=0.3]  fill between[ 
    of = main_SP and bottom_SP
    ];

    \addplot[color = {rgb,1:red,0.00000000;green,0.00000000;blue,1.00000000},
    fill opacity=0.3]  fill between[ 
    of = main_P and upper_P
    ];

    \addplot[color = {rgb,1:red,0.00000000;green,0.00000000;blue,1.00000000},
    fill opacity=0.3]  fill between[ 
    of = main_P and bottom_P
    ];
  \end{axis}

\end{tikzpicture}}\\
  \centering
  \subfloat[$\omega_c = 0.4$]{\begin{tikzpicture}[]
  \begin{axis}[height = {0.4\textwidth}, ylabel = {MSE}, xmin = {0}, xmax = {2}, ymax = {8.0}, xlabel = {$\sigma_v$}, {unbounded coords=jump, xticklabel style={rotate = 0}, yticklabel style={rotate = 0}, ytick = {0.0,2.0,4.0,6.0,8.0}, yticklabels = {0,2,4,6,8},     xshift = 0.0mm,
      yshift = 0.0mm,
      axis background/.style={fill={rgb,1:red,1.00000000;green,1.00000000;blue,1.00000000}}
    }, ymin = {0}, width = {0.505\textwidth}]
    \addplot+ [name path = main_SP,
    color = {rgb,1:red,0.00000000;green,0.50196078;blue,0.00000000},
    draw opacity=1.0,
    line width=1,
    solid,mark = *,
    mark size = 2.0,
    mark options = {
      color = {rgb,1:red,0.00000000;green,0.00000000;blue,0.00000000}, draw opacity = 1.0,
      fill = {rgb,1:red,0.00000000;green,0.50196078;blue,0.00000000}, fill opacity = 1.0,
      line width = 1,
      rotate = 0,
      solid
    },forget plot]coordinates {
      (0.0, 0.016225494314657384)
      (0.1, 0.03384869357619437)
      (0.2, 0.11319049446964957)
      (0.3, 0.3220055068418005)
      (0.4, 0.5367321705118825)
      (0.5, 0.7205865515589556)
      (0.6, 0.8820249887580927)
      (0.7, 1.1635154380902744)
      (0.8, 1.405008316352024)
      (0.9, 1.6764447843244048)
      (1.0, 1.8272700357288227)
      (1.1, 2.426148172712249)
      (1.2, 2.322906472665705)
      (1.3, 2.5227883982194457)
      (1.4, 2.7865782875035334)
      (1.5, 2.886990150364804)
      (1.6, 3.5411846987355826)
      (1.7, 3.070429550751487)
      (1.8, 4.2871220749292505)
      (1.9, 4.271017121648809)
      (2.0, 5.00847299371243)
    };
    \addplot+ [name path = main_P,
    color = {rgb,1:red,0.00000000;green,0.00000000;blue,1.00000000},
    draw opacity=1.0,
    line width=1,
    solid,mark = diamond*,
    mark size = 2.0,
    mark options = {
      color = {rgb,1:red,0.00000000;green,0.00000000;blue,0.00000000}, draw opacity = 1.0,
      fill = {rgb,1:red,0.00000000;green,0.00000000;blue,1.00000000}, fill opacity = 1.0,
      line width = 1,
      rotate = 0,
      solid
    },forget plot]coordinates {
      (0.0, 0.022242322461811896)
      (0.1, 0.07620253211242574)
      (0.2, 0.1850011200034818)
      (0.3, 0.32976452458146144)
      (0.4, 0.38072797485339793)
      (0.5, 0.4109980674707577)
      (0.6, 0.5788109594767163)
      (0.7, 0.6561679424234788)
      (0.8, 0.7891345279342488)
      (0.9, 0.9355364813372158)
      (1.0, 1.0036386289206478)
      (1.1, 0.7971346708246683)
      (1.2, 1.0983180328971922)
      (1.3, 1.081971900975318)
      (1.4, 1.165317346853056)
      (1.5, 1.1939447387864535)
      (1.6, 1.2877185823947965)
      (1.7, 1.166859573724771)
      (1.8, 1.2307561520076942)
      (1.9, 1.5774801470903763)
      (2.0, 1.6879330997370339)
    };
    \addplot+ [name path = upper_SP,
    color = {rgb,1:red,0.00000000;green,0.50196078;blue,0.00000000},
    draw opacity=0.5,
    line width=1,
    solid,mark = none,
    mark size = 2.0,
    mark options = {
      color = {rgb,1:red,0.00000000;green,0.00000000;blue,0.00000000}, draw opacity = 0.4,
      fill = {rgb,1:red,0.00000000;green,0.50196078;blue,0.00000000}, fill opacity = 0.4,
      line width = 1,
      rotate = 0,
      solid
    },forget plot]coordinates {
      (0.0, 0.020588903597037224)
      (0.1, 0.0386291747166934)
      (0.2, 0.14778787215174075)
      (0.3, 0.3543543959623047)
      (0.4, 0.5803715154491409)
      (0.5, 0.787908593961397)
      (0.6, 0.938088455686305)
      (0.7, 1.2288368339921978)
      (0.8, 1.4420841520665475)
      (0.9, 1.8306741212999142)
      (1.0, 2.0954288046851803)
      (1.1, 2.589352771624597)
      (1.2, 2.382405793385434)
      (1.3, 2.8054816434214804)
      (1.4, 3.18858196028225)
      (1.5, 3.375410726838453)
      (1.6, 4.1060485569944145)
      (1.7, 3.5836515012185997)
      (1.8, 5.038072129063668)
      (1.9, 4.990468764096953)
      (2.0, 5.097339937174169)
    };
    \addplot+ [name path = bottom_SP,
    color = {rgb,1:red,0.00000000;green,0.50196078;blue,0.00000000},
    draw opacity=0.5,
    line width=1,
    solid,mark = none,
    mark size = 2.0,
    mark options = {
      color = {rgb,1:red,0.00000000;green,0.00000000;blue,0.00000000}, draw opacity = 0.3,
      fill = {rgb,1:red,0.00000000;green,0.50196078;blue,0.00000000}, fill opacity = 0.3,
      line width = 1,
      rotate = 0,
      solid
    },forget plot]coordinates {
      (0.0, 0.014372810867415799)
      (0.1, 0.023238985847673737)
      (0.2, 0.08494269109083143)
      (0.3, 0.24786843852556162)
      (0.4, 0.4209268385057754)
      (0.5, 0.6216233077741868)
      (0.6, 0.7907567446412759)
      (0.7, 0.9777762906072037)
      (0.8, 1.245032182444875)
      (0.9, 1.4367067187213676)
      (1.0, 1.6804974977356566)
      (1.1, 2.176794667588562)
      (1.2, 2.0708470706179987)
      (1.3, 2.353434069316613)
      (1.4, 2.3716692147340632)
      (1.5, 2.528113711863855)
      (1.6, 3.0749681248711567)
      (1.7, 2.954447078224277)
      (1.8, 2.834932315052023)
      (1.9, 3.261366382304703)
      (2.0, 4.391858108591973)
    };
    \addplot+ [name path = upper_P,
    color = {rgb,1:red,0.00000000;green,0.00000000;blue,1.00000000},
    draw opacity=0.5,
    line width=1,
    solid,mark = none,
    mark size = 2.0,
    mark options = {
      color = {rgb,1:red,0.00000000;green,0.00000000;blue,0.00000000}, draw opacity = 0.4,
      fill = {rgb,1:red,0.00000000;green,0.00000000;blue,1.00000000}, fill opacity = 0.4,
      line width = 1,
      rotate = 0,
      solid
    },forget plot]coordinates {
      (0.0, 0.026227605073744917)
      (0.1, 0.09556502045864818)
      (0.2, 0.2688175134673921)
      (0.3, 0.35924910059127624)
      (0.4, 0.41617576258502487)
      (0.5, 0.5896630364218435)
      (0.6, 0.8315065397562917)
      (0.7, 0.9419111478712081)
      (0.8, 0.9111709995275156)
      (0.9, 1.1033547057927104)
      (1.0, 1.5538034900702593)
      (1.1, 1.0582648314270497)
      (1.2, 1.285571782843168)
      (1.3, 1.545303966685468)
      (1.4, 2.2485177531645326)
      (1.5, 2.825064203602792)
      (1.6, 2.201326957106235)
      (1.7, 2.026661603022141)
      (1.8, 2.76018608309878)
      (1.9, 1.9873869192044507)
      (2.0, 2.931550575186705)
    };
    \addplot+ [name path = bottom_P,
    color = {rgb,1:red,0.00000000;green,0.00000000;blue,1.00000000},
    draw opacity=0.5,
    line width=1,
    solid,mark = none,
    mark size = 2.0,
    mark options = {
      color = {rgb,1:red,0.00000000;green,0.00000000;blue,0.00000000}, draw opacity = 0.3,
      fill = {rgb,1:red,0.00000000;green,0.00000000;blue,1.00000000}, fill opacity = 0.3,
      line width = 1,
      rotate = 0,
      solid
    },forget plot]coordinates {
      (0.0, 0.014501240891213373)
      (0.1, 0.05699922628918366)
      (0.2, 0.1721423978239921)
      (0.3, 0.26977481212178256)
      (0.4, 0.331630487956526)
      (0.5, 0.3853550624518768)
      (0.6, 0.4631007891775907)
      (0.7, 0.544851387839968)
      (0.8, 0.5899559287362344)
      (0.9, 0.8504585887597471)
      (1.0, 0.6919679044913049)
      (1.1, 0.7229052319834822)
      (1.2, 0.8795859849572658)
      (1.3, 0.7565546236867273)
      (1.4, 1.011046788021404)
      (1.5, 1.0630935633794876)
      (1.6, 0.9744127501402784)
      (1.7, 1.0309700017655512)
      (1.8, 1.1329404592261203)
      (1.9, 0.9379884813546714)
      (2.0, 0.9776232023502199)
    };
    
    \addplot[color = {rgb,1:red,0.00000000;green,0.50196078;blue,0.00000000},
    fill opacity=0.3]  fill between[ 
    of = main_SP and upper_SP
    ];

    \addplot[color = {rgb,1:red,0.00000000;green,0.50196078;blue,0.00000000},
    fill opacity=0.3]  fill between[ 
    of = main_SP and bottom_SP
    ];

    \addplot[color = {rgb,1:red,0.00000000;green,0.00000000;blue,1.00000000},
    fill opacity=0.3]  fill between[ 
    of = main_P and upper_P
    ];

    \addplot[color = {rgb,1:red,0.00000000;green,0.00000000;blue,1.00000000},
    fill opacity=0.3]  fill between[ 
    of = main_P and bottom_P
    ];
  \end{axis}

\end{tikzpicture}}
  \subfloat[$\omega_c = 0.2$]{\begin{tikzpicture}[]
  \begin{axis}[height = {0.4\textwidth}, ylabel = {}, xmin = {0}, xmax = {2}, ymax = {8.0}, xlabel = {$\sigma_v$}, {unbounded coords=jump, xticklabel style={rotate = 0}, yticklabel style={rotate = 0}, ytick = {}, yticklabels = {},     xshift = 0.0mm,
      yshift = 0.0mm,
      axis background/.style={fill={rgb,1:red,1.00000000;green,1.00000000;blue,1.00000000}}
    }, ymin = {0}, width = {0.495\textwidth}]
    \addplot+ [name path = main_SP,
    color = {rgb,1:red,0.00000000;green,0.50196078;blue,0.00000000},
    draw opacity=1.0,
    line width=1,
    solid,mark = *,
    mark size = 2.0,
    mark options = {
      color = {rgb,1:red,0.00000000;green,0.00000000;blue,0.00000000}, draw opacity = 1.0,
      fill = {rgb,1:red,0.00000000;green,0.50196078;blue,0.00000000}, fill opacity = 1.0,
      line width = 1,
      rotate = 0,
      solid
    },forget plot]coordinates {
      (0.0, 0.016225494314657384)
      (0.1, 0.032660023414527894)
      (0.2, 0.14578419590231395)
      (0.3, 0.37963458176701065)
      (0.4, 0.6936502215849838)
      (0.5, 0.9899928595191068)
      (0.6, 1.2281716742708966)
      (0.7, 1.4831202717863685)
      (0.8, 1.8127890338932766)
      (0.9, 2.1221473094899785)
      (1.0, 2.661965844642107)
      (1.1, 2.835061006516824)
      (1.2, 3.331886907684952)
      (1.3, 4.010206787187191)
      (1.4, 4.639477451159857)
      (1.5, 4.856145291024912)
      (1.6, 5.694832643360577)
      (1.7, 5.998766681733416)
      (1.8, 6.697789327889121)
      (1.9, 7.150523255595445)
      (2.0, 8.10581297045607)
    };
    \addplot+ [name path = main_P,
    color = {rgb,1:red,0.00000000;green,0.00000000;blue,1.00000000},
    draw opacity=1.0,
    line width=1,
    solid,mark = diamond*,
    mark size = 2.0,
    mark options = {
      color = {rgb,1:red,0.00000000;green,0.00000000;blue,0.00000000}, draw opacity = 1.0,
      fill = {rgb,1:red,0.00000000;green,0.00000000;blue,1.00000000}, fill opacity = 1.0,
      line width = 1,
      rotate = 0,
      solid
    },forget plot]coordinates {
      (0.0, 0.022242322461811896)
      (0.1, 0.06793219073545545)
      (0.2, 0.137986712746754)
      (0.3, 0.24499224550549625)
      (0.4, 0.37510732410823266)
      (0.5, 0.48302413298503666)
      (0.6, 0.7440780816151722)
      (0.7, 0.8130314999889201)
      (0.8, 0.6961034312493974)
      (0.9, 0.8295792091341392)
      (1.0, 1.0247877108255485)
      (1.1, 1.0749550854259797)
      (1.2, 1.0196634892904346)
      (1.3, 1.5611312886761213)
      (1.4, 1.6253164986243613)
      (1.5, 1.4685108720302136)
      (1.6, 1.4059590955570909)
      (1.7, 1.3950996948781085)
      (1.8, 1.38632655517115)
      (1.9, 1.5804411878907483)
      (2.0, 1.5642438117147603)
    };
    \addplot+ [name path = upper_SP,
    color = {rgb,1:red,0.00000000;green,0.50196078;blue,0.00000000},
    draw opacity=0.5,
    line width=1,
    solid,mark = none,
    mark size = 2.0,
    mark options = {
      color = {rgb,1:red,0.00000000;green,0.00000000;blue,0.00000000}, draw opacity = 0.4,
      fill = {rgb,1:red,0.00000000;green,0.50196078;blue,0.00000000}, fill opacity = 0.4,
      line width = 1,
      rotate = 0,
      solid
    },forget plot]coordinates {
      (0.0, 0.020588903597037224)
      (0.1, 0.043265518839364074)
      (0.2, 0.19919896204743037)
      (0.3, 0.42779088015232186)
      (0.4, 0.7449231057548479)
      (0.5, 1.0825556025396141)
      (0.6, 1.2775017075383992)
      (0.7, 1.6498463778148431)
      (0.8, 2.0281683709943974)
      (0.9, 2.4082741419816105)
      (1.0, 2.845067658068115)
      (1.1, 3.0702623206558877)
      (1.2, 3.686126343173285)
      (1.3, 4.428676377335429)
      (1.4, 4.82607140068889)
      (1.5, 5.650443551800421)
      (1.6, 6.226287882995058)
      (1.7, 6.684675202574717)
      (1.8, 7.588826254196458)
      (1.9, 7.786328281111111)
      (2.0, 8.565609073068297)
    };
    \addplot+ [name path = bottom_SP,
    color = {rgb,1:red,0.00000000;green,0.50196078;blue,0.00000000},
    draw opacity=0.5,
    line width=1,
    solid,mark = none,
    mark size = 2.0,
    mark options = {
      color = {rgb,1:red,0.00000000;green,0.00000000;blue,0.00000000}, draw opacity = 0.3,
      fill = {rgb,1:red,0.00000000;green,0.50196078;blue,0.00000000}, fill opacity = 0.3,
      line width = 1,
      rotate = 0,
      solid
    },forget plot]coordinates {
      (0.0, 0.014372810867415799)
      (0.1, 0.027969022354768946)
      (0.2, 0.12111425368862713)
      (0.3, 0.350953823322691)
      (0.4, 0.567009724019513)
      (0.5, 0.8466684643457884)
      (0.6, 1.0813366839694394)
      (0.7, 1.318023708290133)
      (0.8, 1.6772819508474104)
      (0.9, 1.9410063173288499)
      (1.0, 2.2455324197013833)
      (1.1, 2.65419813873693)
      (1.2, 3.2069573607882083)
      (1.3, 3.5116823749713313)
      (1.4, 4.179280302811249)
      (1.5, 4.614296043828699)
      (1.6, 5.265714018243033)
      (1.7, 5.376170594695201)
      (1.8, 5.97421150698121)
      (1.9, 6.521272701561363)
      (2.0, 7.411556570624057)
    };
    \addplot+ [name path = upper_P,
    color = {rgb,1:red,0.00000000;green,0.00000000;blue,1.00000000},
    draw opacity=0.5,
    line width=1,
    solid,mark = none,
    mark size = 2.0,
    mark options = {
      color = {rgb,1:red,0.00000000;green,0.00000000;blue,0.00000000}, draw opacity = 0.4,
      fill = {rgb,1:red,0.00000000;green,0.00000000;blue,1.00000000}, fill opacity = 0.4,
      line width = 1,
      rotate = 0,
      solid
    },forget plot]coordinates {
      (0.0, 0.026227605073744917)
      (0.1, 0.08372977392084405)
      (0.2, 0.1499478217516669)
      (0.3, 0.2922836047173133)
      (0.4, 0.387609165146327)
      (0.5, 0.5222607716337279)
      (0.6, 0.797909662223172)
      (0.7, 1.0128084603597876)
      (0.8, 0.9712334907272814)
      (0.9, 1.0866606221876443)
      (1.0, 1.732394072938897)
      (1.1, 1.45331061500545)
      (1.2, 1.3040657417012245)
      (1.3, 2.0260488108964525)
      (1.4, 1.957703217851834)
      (1.5, 2.1590541168766286)
      (1.6, 2.2020346745097603)
      (1.7, 2.157674482958827)
      (1.8, 2.1469777885227885)
      (1.9, 2.1249689012854738)
      (2.0, 2.1424890040916513)
    };
    \addplot+ [name path = bottom_P,
    color = {rgb,1:red,0.00000000;green,0.00000000;blue,1.00000000},
    draw opacity=0.5,
    line width=1,
    solid,mark = none,
    mark size = 2.0,
    mark options = {
      color = {rgb,1:red,0.00000000;green,0.00000000;blue,0.00000000}, draw opacity = 0.3,
      fill = {rgb,1:red,0.00000000;green,0.00000000;blue,1.00000000}, fill opacity = 0.3,
      line width = 1,
      rotate = 0,
      solid
    },forget plot]coordinates {
      (0.0, 0.014501240891213373)
      (0.1, 0.059408982914894465)
      (0.2, 0.11002874628002443)
      (0.3, 0.21004322449808174)
      (0.4, 0.343023332722028)
      (0.5, 0.4102069821642181)
      (0.6, 0.49739954696894434)
      (0.7, 0.5250599494652574)
      (0.8, 0.5239112235309031)
      (0.9, 0.6925518428090753)
      (1.0, 0.7677216956868528)
      (1.1, 0.7904286237508914)
      (1.2, 0.8193204089498858)
      (1.3, 0.9753873170203295)
      (1.4, 1.472661171669739)
      (1.5, 1.05678020195855)
      (1.6, 1.2021771238237358)
      (1.7, 1.1152007913840594)
      (1.8, 0.8386258349401835)
      (1.9, 1.255761878961675)
      (2.0, 1.4765470953700683)
    };

    \addplot[color = {rgb,1:red,0.00000000;green,0.50196078;blue,0.00000000},
    fill opacity=0.3]  fill between[ 
    of = main_SP and upper_SP
    ];

    \addplot[color = {rgb,1:red,0.00000000;green,0.50196078;blue,0.00000000},
    fill opacity=0.3]  fill between[ 
    of = main_SP and bottom_SP
    ];

    \addplot[color = {rgb,1:red,0.00000000;green,0.00000000;blue,1.00000000},
    fill opacity=0.3]  fill between[ 
    of = main_P and upper_P
    ];

    \addplot[color = {rgb,1:red,0.00000000;green,0.00000000;blue,1.00000000},
    fill opacity=0.3]  fill between[ 
    of = main_P and bottom_P
    ];
  \end{axis}

\end{tikzpicture}}
  \caption{Free-run simulation MSE over the validation window  \textit{vs}
    standard deviation of colored equation error.  The output error $w=0$ and
    the equation error $v$ is a colored Gaussian noise obtained by applying
    a 4th-order lowpass Butterworth filter with cutoff frequency $\omega_c$
    to white Gaussian noise in both the forward and reverse directions.  The figure
    shows the result for different values of $\omega_c$, where $\omega_c$
    is the normalized frequency (with $\omega_c = 1$ the Nyquist frequency).
    The main line indicates the median and the shaded region indicates interquartile range.
    These statistics were computed from 12 realizations.}
  \label{fig:mse_vs_noiselevel_colored}
\end{figure}
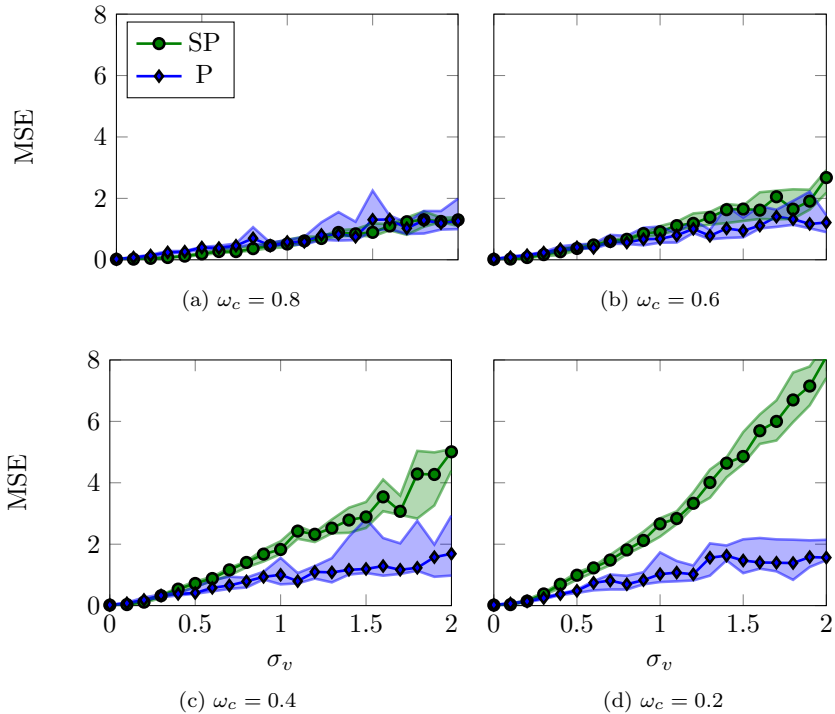

Equation error can be interpreted as the effect of unaccounted inputs and
unmodeled dynamics, hence situations where this error is not auto-correlated
are very unlikely. Therefore, the only situations we found series-parallel training
to perform better (when the equation error power spectral density occupy almost
the whole frequency spectrum)  seem unlikely to appear in practice.
This may justify parallel training to produce better
models for real application problems as the pilot plant in Example 1,
the battery modeling described in~\cite{zhang_new_2014}, or the boiler unit
in~\cite{patan_nonlinear_2012}.

\begin{table}[t]
  \centering
  \caption{Free-run simulation MSE on the validation window for parallel and series-parallel
    training. Both the mean and
    the standard deviation are displayed (30\% trimmed estimation computed from 12 realizations).
    In situation \textbf{(a)}, the training data was generated with
    zero output error ($w = 0$) and $v$ is a
    Gaussian random process. In \textbf{(b)},  the training data was
    generated with $v = 0$ and $w$ is a
    Gaussian random process. The Gaussian random process has
    standard deviation $\sigma = 1.0$ and power spectral density
    confined to the given frequency band.
    In both situations, the rows where the frequency ranges from $0.0$
    to $1.0$ (the whole spectrum) corresponds to white noise, in the remaining
    rows we apply a 4th-order lowpass (or highpass) Butterworth filter
    to white Gaussian noise (in both the forward and reverse directions)
    in order to obtain the signal in the desired frequency band.
    The cell of the training method with the best validation results between
    the two models is colored.
    Its colored \colorbox{red!25}{red} when the difference in the MSE is larger than the
    sum of standard deviations and \colorbox{yellow!25}{yellow} when it is not.
    }
  \begin{tabular}{|c|c|c||c|c|}
    \cline{2-5}
    \multicolumn{1}{c|}{}& \multicolumn{2}{c||}{\textbf{(a)}~~$v\not= 0$, $w = 0$}& \multicolumn{2}{c|}{\textbf{(b)}~~$w\not= 0$, $v = 0$} \\
    \hline
    $\omega$ range& SP & P & SP & P\\
    \hline
    $0.0\rightarrow 1.0$ & \cellcolor{red!25}$0.36 \pm 0.09$ & $0.78 \pm 0.18$ & $0.58 \pm 0.03$ & \cellcolor{red!25}$0.13 \pm 0.03$\\
    \hline
    $0.0\rightarrow 0.8$ & \cellcolor{yellow!25}$0.53 \pm 0.13$ & $0.58 \pm 0.10$ & $0.44 \pm 0.02$ & \cellcolor{red!25}$0.12 \pm 0.03$\\
    \hline
    $0.0\rightarrow 0.6$ & $0.94 \pm 0.13$ & \cellcolor{yellow!25}$0.75 \pm 0.23$ & $0.30 \pm 0.05$ & \cellcolor{red!25}$0.15 \pm 0.05$\\
    \hline
    $0.0\rightarrow 0.4$ & $1.86 \pm 0.20$ & \cellcolor{red!25}$1.07 \pm 0.38$ & $0.46 \pm 0.05$ & \cellcolor{red!25}$0.15 \pm 0.02$\\
    \hline
    $0.0\rightarrow 0.2$ & $2.60 \pm 0.26$ & \cellcolor{red!25}$1.16 \pm 0.44$ & $0.71 \pm 0.05$ & \cellcolor{red!25}$0.18 \pm 0.04$\\
    \hline
    \hline
    $0.0\rightarrow 1.0$ &  \cellcolor{red!25} $0.36 \pm 0.09$ & $ 0.78 \pm 0.18$ & $0.58 \pm 0.03$ &  \cellcolor{red!25}$0.13 \pm 0.03$\\
    \hline
    $0.2\rightarrow 1.0$ &  \cellcolor{yellow!25} $0.52 \pm 0.07$ & $0.57 \pm 0.12$ & $0.59 \pm 0.03$ &  \cellcolor{red!25}$0.09 \pm 0.02$\\
    \hline
    $0.4\rightarrow 1.0$ & $0.57 \pm 0.08$ &  \cellcolor{red!25}$0.22 \pm 0.05$ & $0.63 \pm 0.05$ &  \cellcolor{red!25}$0.05 \pm 0.01$\\
    \hline
    $0.6\rightarrow 1.0$ & $0.54 \pm 0.07$ &  \cellcolor{red!25}$0.21 \pm 0.01$ & $0.68 \pm 0.03$ &  \cellcolor{red!25}$0.03 \pm 0.02$\\
    \hline
    $0.8\rightarrow 1.0$ & $0.58 \pm 0.05$ &  \cellcolor{red!25}$0.17 \pm 0.04$ & $0.78 \pm 0.09$ &  \cellcolor{red!25}$0.03 \pm 0.01$\\
    \hline
  \end{tabular}
  \label{tab:filtered_noise}
\end{table}

\subsection{Running time}

In Section~\ref{sec:complexity-analysis} we find out the computational
complexity of $\mathcal{O}(N\cdot N_\Theta^2+N_\Theta^3)$. The first term
$\mathcal{O}(N\cdot N_\Theta^2)$
seems to dominate and in Figure~\ref{fig:timmings} we
show that the running time grows linearly with the number of training samples $N$
and quadratically with the number of parameters $N_\Theta$.

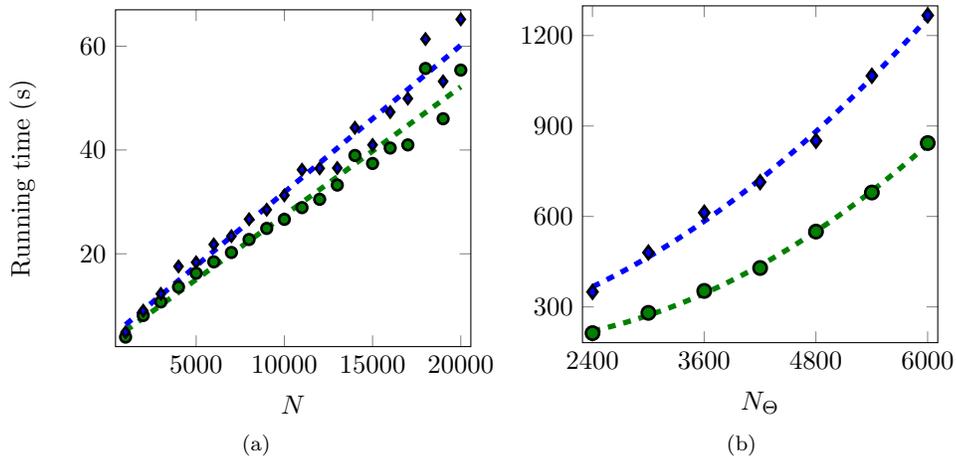
\begin{figure}[ht]
  \centering
  \subfloat[]{\begin{tikzpicture}[]
\begin{axis}[height = {0.5\textwidth}, ylabel = {Running time (s)}, xmin = {0.43000000000000005}, xmax = {20.57}, ymax = {67.015840160528}, xlabel = {$N$}, {unbounded coords=jump, xticklabel style={rotate = 0}, xtick = {5,10,15,20}, xticklabels = {5000,10000,15000,20000}, yticklabel style={rotate = 0}, ytick = {20.0,40.0,60.0}, yticklabels = {20,40,60},     xshift = 0.0mm,
    yshift = 0.0mm,
    axis background/.style={fill={rgb,1:red,1.00000000;green,1.00000000;blue,1.00000000}}
}, ymin = {2.1119371098720006}, width = {0.52\textwidth}]\addplot+ [color = {rgb,1:red,0.00000000;green,0.50196078;blue,0.00000000},
draw opacity=1.0,
line width=2,
dashed,mark = none,
mark size = 2.0,
mark options = {
    color = {rgb,1:red,0.00000000;green,0.00000000;blue,0.00000000}, draw opacity = 1.0,
    fill = {rgb,1:red,0.00000000;green,0.50196078;blue,0.00000000}, fill opacity = 1.0,
    line width = 1,
    rotate = 0,
    solid
},forget plot]coordinates {
(1.0, 5.1450640698657)
(2.0, 7.619901755706154)
(3.0, 10.094739441546606)
(4.0, 12.56957712738706)
(5.0, 15.044414813227512)
(6.0, 17.519252499067964)
(7.0, 19.99409018490842)
(8.0, 22.46892787074887)
(9.0, 24.943765556589323)
(10.0, 27.418603242429775)
(11.0, 29.893440928270227)
(12.0, 32.36827861411068)
(13.0, 34.843116299951134)
(14.0, 37.31795398579159)
(15.0, 39.79279167163204)
(16.0, 42.267629357472494)
(17.0, 44.74246704331294)
(18.0, 47.2173047291534)
(19.0, 49.69214241499385)
(20.0, 52.1669801008343)
};
\addplot+ [only marks = {true}, color = {rgb,1:red,0.88887350;green,0.43564919;blue,0.27812294},
draw opacity=1.0,
line width=0,
solid,mark = *,
mark size = 2.0,
mark options = {
    color = {rgb,1:red,0.00000000;green,0.00000000;blue,0.00000000}, draw opacity = 1.0,
    fill = {rgb,1:red,0.00000000;green,0.50196078;blue,0.00000000}, fill opacity = 1.0,
    line width = 1,
    rotate = 0,
    solid
},forget plot]coordinates {
(1.0, 3.9488400264)
(2.0, 8.1103074454)
(3.0, 10.7230784358)
(4.0, 13.605025957200002)
(5.0, 16.245763322400002)
(6.0, 18.445338557)
(7.0, 20.25055617)
(8.0, 22.756445266)
(9.0, 24.891539415400004)
(10.0, 26.636462808200005)
(11.0, 28.897779038800003)
(12.0, 30.501576394200004)
(13.0, 33.2427330612)
(14.0, 38.950694721999994)
(15.0, 37.406495942199996)
(16.0, 40.38032011040001)
(17.0, 40.9849578504)
(18.0, 55.7154669014)
(19.0, 46.0280774328)
(20.0, 55.3989828498)
};
\addplot+ [color = {rgb,1:red,0.00000000;green,0.00000000;blue,1.00000000},
draw opacity=1.0,
line width=2,
dashed,mark = none,
mark size = 2.0,
mark options = {
    color = {rgb,1:red,0.00000000;green,0.00000000;blue,0.00000000}, draw opacity = 1.0,
    fill = {rgb,1:red,0.00000000;green,0.00000000;blue,1.00000000}, fill opacity = 1.0,
    line width = 1,
    rotate = 0,
    solid
},forget plot]coordinates {
(1.0, 6.339855907608542)
(2.0, 9.17378019510238)
(3.0, 12.007704482596216)
(4.0, 14.841628770090054)
(5.0, 17.67555305758389)
(6.0, 20.509477345077727)
(7.0, 23.343401632571563)
(8.0, 26.177325920065403)
(9.0, 29.01125020755924)
(10.0, 31.845174495053076)
(11.0, 34.679098782546916)
(12.0, 37.513023070040745)
(13.0, 40.34694735753459)
(14.0, 43.18087164502842)
(15.0, 46.01479593252226)
(16.0, 48.848720220016105)
(17.0, 51.682644507509934)
(18.0, 54.51656879500378)
(19.0, 57.35049308249761)
(20.0, 60.18441736999145)
};
\addplot+ [only marks = {true}, color = {rgb,1:red,0.76444018;green,0.44411178;blue,0.82429754},
draw opacity=1.0,
line width=0,
solid,mark = diamond*,
mark size = 2.0,
mark options = {
    color = {rgb,1:red,0.00000000;green,0.00000000;blue,0.00000000}, draw opacity = 1.0,
    fill = {rgb,1:red,0.00000000;green,0.00000000;blue,1.00000000}, fill opacity = 1.0,
    line width = 1,
    rotate = 0,
    solid
},forget plot]coordinates {
(1.0, 4.8995067286000005)
(2.0, 9.0368002918)
(3.0, 12.2782962954)
(4.0, 17.564293835799997)
(5.0, 18.3593362458)
(6.0, 21.803271315200003)
(7.0, 23.402987067)
(8.0, 26.6597129344)
(9.0, 28.491185781600002)
(10.0, 31.265131535600002)
(11.0, 36.2009485852)
(12.0, 36.467519681400006)
(13.0, 36.5322908012)
(14.0, 44.287468380600004)
(15.0, 40.984675576200004)
(16.0, 47.3055879866)
(17.0, 49.9178021274)
(18.0, 61.3822933792)
(19.0, 53.224686983000005)
(20.0, 65.178937244)
};
\end{axis}

\end{tikzpicture}}
  \subfloat[]{\begin{tikzpicture}[]
\begin{axis}[height = {0.5\textwidth}, ylabel = {}, xmin = {2293.0}, xmax = {6109.0}, ymax = {1298.5239251613782}, xlabel = {$N_\Theta$}, {unbounded coords=jump, xticklabel style={rotate = 0}, xtick = {2400.0,3600.0,4800.0,6000.0}, xticklabels = {2400,3600,4800,6000}, yticklabel style={rotate = 0}, ytick = {300.0,600.0,900.0,1200.0}, yticklabels = {300,600,900,1200},     xshift = 0.0mm,
    yshift = 0.0mm,
    axis background/.style={fill={rgb,1:red,1.00000000;green,1.00000000;blue,1.00000000}}
}, ymin = {181.391068506422}, width = {0.52\textwidth}]\addplot+ [color = {rgb,1:red,0.00000000;green,0.50196078;blue,0.00000000},
draw opacity=1.0,
line width=2,
dashed,mark = none,
mark size = 2.0,
mark options = {
    color = {rgb,1:red,0.00000000;green,0.00000000;blue,0.00000000}, draw opacity = 1.0,
    fill = {rgb,1:red,0.00000000;green,0.50196078;blue,0.00000000}, fill opacity = 1.0,
    line width = 1,
    rotate = 0,
    solid
},forget plot]coordinates {
(2401.0, 220.0738484320665)
(2461.0, 224.35986053945172)
(2521.0, 228.850024983393)
(2581.0, 233.54434176389037)
(2641.0, 238.44281088094385)
(2701.0, 243.54543233455342)
(2761.0, 248.8522061247191)
(2821.0, 254.36313225144087)
(2881.0, 260.0782107147187)
(2941.0, 265.9974415145527)
(3001.0, 272.1208246509427)
(3061.0, 278.44836012388885)
(3121.0, 284.9800479333911)
(3181.0, 291.7158880794494)
(3241.0, 298.6558805620639)
(3301.0, 305.8000253812344)
(3361.0, 313.14832253696096)
(3421.0, 320.7007720292437)
(3481.0, 328.4573738580825)
(3541.0, 336.4181280234775)
(3601.0, 344.5830345254284)
(3661.0, 352.95209336393555)
(3721.0, 361.52530453899874)
(3781.0, 370.30266805061797)
(3841.0, 379.2841838987934)
(3901.0, 388.46985208352487)
(3961.0, 397.85967260481243)
(4021.0, 407.45364546265614)
(4081.0, 417.2517706570559)
(4141.0, 427.2540481880117)
(4201.0, 437.4604780555237)
(4261.0, 447.8710602595918)
(4321.0, 458.4857948002159)
(4381.0, 469.3046816773961)
(4441.0, 480.32772089113246)
(4501.0, 491.55491244142496)
(4561.0, 502.9862563282734)
(4621.0, 514.6217525516781)
(4681.0, 526.4614011116388)
(4741.0, 538.5052020081555)
(4801.0, 550.7531552412285)
(4861.0, 563.2052608108575)
(4921.0, 575.8615187170426)
(4981.0, 588.7219289597838)
(5041.0, 601.7864915390811)
(5101.0, 615.0552064549345)
(5161.0, 628.528073707344)
(5221.0, 642.2050932963095)
(5281.0, 656.0862652218311)
(5341.0, 670.171589483909)
(5401.0, 684.4610660825429)
(5461.0, 698.9546950177328)
(5521.0, 713.6524762894787)
(5581.0, 728.5544098977809)
(5641.0, 743.6604958426392)
(5701.0, 758.9707341240535)
(5761.0, 774.4851247420239)
(5821.0, 790.2036676965505)
(5881.0, 806.1263629876332)
(5941.0, 822.2532106152718)
(6001.0, 838.5842105794666)
};
\addplot+ [only marks = {true}, color = {rgb,1:red,0.88887350;green,0.43564919;blue,0.27812294},
draw opacity=1.0,
line width=0,
dashed,mark = *,
mark size = 2.5,
mark options = {
    color = {rgb,1:red,0.00000000;green,0.00000000;blue,0.00000000}, draw opacity = 1.0,
    fill = {rgb,1:red,0.00000000;green,0.50196078;blue,0.00000000}, fill opacity = 1.0,
    line width = 1,
    rotate = 0,
    solid
},forget plot]coordinates {
(2401.0, 213.0080361476)
(3001.0, 280.24508134620004)
(3601.0, 353.0498862742)
(4201.0, 429.0901231712)
(4801.0, 549.807726716)
(5401.0, 679.3691710214)
(6001.0, 843.4665928906002)
};
\addplot+ [color = {rgb,1:red,0.00000000;green,0.00000000;blue,1.00000000},
draw opacity=1.0,
line width=2,
dashed,mark = none,
mark size = 2.0,
mark options = {
    color = {rgb,1:red,0.00000000;green,0.00000000;blue,0.00000000}, draw opacity = 1.0,
    fill = {rgb,1:red,0.00000000;green,0.00000000;blue,1.00000000}, fill opacity = 1.0,
    line width = 1,
    rotate = 0,
    solid
},forget plot]coordinates {
(2401.0, 366.87265628949064)
(2461.0, 375.85246503915096)
(2521.0, 385.0317606135293)
(2581.0, 394.41054301262545)
(2641.0, 403.98881223643957)
(2701.0, 413.7665682849715)
(2761.0, 423.74381115822143)
(2821.0, 433.92054085618923)
(2881.0, 444.29675737887493)
(2941.0, 454.8724607262785)
(3001.0, 465.64765089840006)
(3061.0, 476.6223278952395)
(3121.0, 487.79649171679677)
(3181.0, 499.1701423630721)
(3241.0, 510.7432798340652)
(3301.0, 522.5159041297762)
(3361.0, 534.4880152502052)
(3421.0, 546.659613195352)
(3481.0, 559.0306979652169)
(3541.0, 571.6012695597994)
(3601.0, 584.3713279791)
(3661.0, 597.3408732231185)
(3721.0, 610.5099052918549)
(3781.0, 623.8784241853091)
(3841.0, 637.4464299034813)
(3901.0, 651.2139224463714)
(3961.0, 665.1809018139794)
(4021.0, 679.3473680063053)
(4081.0, 693.7133210233492)
(4141.0, 708.2787608651108)
(4201.0, 723.0436875315905)
(4261.0, 738.008101022788)
(4321.0, 753.1720013387035)
(4381.0, 768.5353884793368)
(4441.0, 784.098262444688)
(4501.0, 799.8606232347571)
(4561.0, 815.8224708495442)
(4621.0, 831.9838052890491)
(4681.0, 848.344626553272)
(4741.0, 864.9049346422128)
(4801.0, 881.6647295558714)
(4861.0, 898.6240112942479)
(4921.0, 915.7827798573425)
(4981.0, 933.1410352451549)
(5041.0, 950.6987774576852)
(5101.0, 968.4560064949334)
(5161.0, 986.4127223568995)
(5221.0, 1004.5689250435835)
(5281.0, 1022.9246145549854)
(5341.0, 1041.479790891105)
(5401.0, 1060.234454051943)
(5461.0, 1079.1886040374984)
(5521.0, 1098.342240847772)
(5581.0, 1117.6953644827636)
(5641.0, 1137.2479749424729)
(5701.0, 1157.0000722269)
(5761.0, 1176.9516563360453)
(5821.0, 1197.1027272699084)
(5881.0, 1217.453285028489)
(5941.0, 1238.003329611788)
(6001.0, 1258.7528610198046)
};
\addplot+ [only marks = {true}, color = {rgb,1:red,0.76444018;green,0.44411178;blue,0.82429754},
draw opacity=1.0,
line width=0,
dashed,mark = diamond*,
mark size = 2.5,
mark options = {
    color = {rgb,1:red,0.00000000;green,0.00000000;blue,0.00000000}, draw opacity = 1.0,
    fill = {rgb,1:red,0.00000000;green,0.00000000;blue,1.00000000}, fill opacity = 1.0,
    line width = 1,
    rotate = 0,
    solid
},forget plot]coordinates {
(2401.0, 349.78662099679997)
(3001.0, 480.2152939866)
(3601.0, 612.5294387638)
(4201.0, 713.8697530748001)
(4801.0, 850.8519667264)
(5401.0, 1066.4273362576002)
(6001.0, 1266.9069575202002)
};
\end{axis}

\end{tikzpicture}}
  \caption{Running time (in seconds) of training (100 epochs).
    The average running time of 5 realizations is displayed
    for series-parallel training (\protect\greencircle)
    and for parallel training (\protect\bluediamond). The
    neural network has a single hidden layer and a total of $N_\Theta$ parameters
    to be estimated. The training set has $N$ samples and was generated
    as described in Figure~\ref{fig:simulation_val}. In (a), we fix $N_\Theta = 61$
    and plot the timings as a function of the number of training samples $N$,
    a line is adjusted to illustrate the running time grows linearly
    with the training size for both training methods. In (b), we fix $N = 10000$
    and plot the timings as a function of the number of parameters $N_\Theta$,
    a second order polynomial is adjusted to illustrate the running time quadratic growth.}
  \label{fig:timmings}
\end{figure}

The running time growing with the same rate for both training methods
implies that the ratio between series-parallel and parallel training running time is
bounded by a constant. For the examples we presented in this paper
the parallel training takes about 20\% longer than series-parallel
training. Hence the difference of running times for sequential
execution does not justify the use of one method over the other.

\section{Discussion}
\label{sec:discussion}

\subsection{Convergence towards a local minima}

The optimization problem that results from both series-parallel
and parallel training of neural networks are non-convex and may
have multiple local minima. The solution of the Levenberg-Marquardt algorithm
discussed in this paper, as well as most
algorithms of interest for training neural networks (e.g. stochastic gradient
descent, L-BFGS, conjugate gradient), converges towards a local minimum\footnote{
  It is proved in~\cite{more_levenberg-marquardt_1978} that the
  Levenberg-Marquardt (not exactly the one discussed here) converges towards
  a local minimum or a stationary point under simple assumptions.}.
However, there is no guarantee for neither series-parallel
nor parallel training that the solution found is the global optimum. The convergence to
``poor'' local solutions may happen for both training methods, however,
as illustrated in the numerical examples, it seems to happen more often
for parallel training.

For the examples presented in this paper the possibility of being trapped
in ``poor'' local solutions is only a small inconvenience, requiring
the data to be carefully normalized and, in some rare situations,
the model to be retrained. 
An exception are chaotic systems for which small variations
in the parameters may cause great variations in the free-run simulation trajectory,
causing the parallel training objective function
to be very intricate and full of undesirable local minima.

\subsection{Signal unboundedness during training}

Signals obtained in intermediary steps of parallel
identification may become unbounded. 
The one-step-ahead predictor, used in series-parallel training,
is always bounded since the prediction depends
only on measured values -- it has a FIR (Finite Impulse Response) structure --
while, for the parallel training, the  free-run simulation could be unbounded
for some choice of parameters because of its dependence on its own
past simulation values.

This is a potential problem because during an intermediary stage of
parallel training a choice of $\mathbf{\Theta}^k$ that results in
unbounded values of $\hat{\mathbf{y}}_s[k]$ may need to be evaluated,
resulting in overflows. Hence, optimization
algorithms that may work well minimizing one-step-ahead prediction
errors $\mathbf{e}_1$, may fail when minimizing simulation
errors $\mathbf{e}_s$.

For instance, steepest descent algorithms with a fixed step size may,
for a poor choice of step size, fall into a region in the parameter space
for which the signals are unbounded and the computation of the gradient
will suffer from overflow and prevent the algorithm from continuing (since it does not have
a direction to follow).
This may also happen in more practical line search algorithms (e.g. the one
described in~\cite[Sec. 3.5]{nocedal_numerical_2006}).

The Levenberg-Marquardt algorithm, on the other hand,
is robust against this kind of problem because every step
$\mathbf{\mathbf{\Theta}}^n + \mathbf{\Delta}\mathbf{\mathbf{\Theta}}^n$
that causes overflow in the objective function computation yields a negative
$\rho_n$\footnote{
  Programming languages as Matlab, C, C++ and  Julia
  return the floating point value encoded for infinity when
  an overflow occur. In this case formula~(\ref{eq:reduction_ratio})
  yields a negative $\rho_n$.},
hence the step is rejected by the algorithm and $\lambda_n$ is increased.
The increase in $\lambda_n$ causes the length of 
$\mathbf{\Delta}\mathbf{\mathbf{\Theta}}^n$ to decrease\footnote{This inverse relation
  between $\lambda_n$ and $\|\mathbf{\Delta}\mathbf{\mathbf{\Theta}}^n\|$
  is explained in~\cite{nocedal_numerical_2006}.}.
Therefore, the step length is decreased until a point is found sufficiently close to the current
iteration such that overflow does not occur. Hence,
the Levenberg-Marquardt algorithm does not fail or stall due to overflows.
Similar reasoning could be used for any trust-region method or for backtracking line-search.

Regardless of the optimization algorithm, signal unboundedness is not a problem for feedforward networks
with bounded activation functions (e.g. Logistic or Hyperbolic Tangent) in the hidden layers,
because its output is always bounded. Hence parallel training of these particular neural networks
is not affected by the previously mentioned difficulties.

\subsection{Time series prediction}

  It is important to highlight that parallel training is inappropriate to train neural networks for predicting time series
  in general. That is, parallel configuration is generally inadequate for the case when there are no inputs ($N_u = 0$).

  For any asymptotically stable system in parallel configuration, the absence of inputs would
  make the free-run simulation converge towards an equilibrium, and even models that should be capable of providing
  good predictions for a few steps-ahead in the time series, may present poor performance for the entire training window.
  Hence, minimizing $\|\mathbf{e}_s\|^2$ will not provide good results in general.

  It still make sense to use series-parallel training for time series models. That
  is because, feeding measured values into the neural network keeps it from converging
  towards zero (for asymptotically stable systems) and makes the estimation robust against unknown
  disturbances affecting the time series. An interesting approach that mixes parallel and series-parallel
  training for time series prediction is given in~\cite{menezes_long-term_2008}.

\subsection{Batch vs online training}

We have based our analysis on batch training.
However, instead of using all the available samples for training
at once we could have fed samples one-by-one or chunk-by-chunk
to the training algorithm (online training).
The choice of parallel or  series-parallel training, however, is orthogonal
to the choice between online and batch training. And most of the ideas presented
in this paper, including: 1) the unified framework; 2) the discussion about poor local minima;
and 3) the study of how colored noise affects the parameter
estimation;  are all applicable to the case of online training.

\subsection{Parallelization}

For the examples presented here the difference of running times for sequential
execution does not justify the use of one method over the other.
Furthermore, both methods have the same time complexity.
Parallel training is, however, much less amenable to parallelization
because of the dependencies introduced by the recursive relations used for computing
its error and Jacobian matrix.

\section{Conclusion and future work}
\label{sec:concl-future-work}

In this paper we have studied different aspects of parallel training
under a nonlinear least squares formulation. Several published works
take for granted that series-parallel training always provides better
results with lower computational cost. The results presented
in this paper show that \textit{this is not always the case}
and that parallel training does provide advantages that justify
its use in several situations.
The results presented in the numerical examples suggest
parallel training can provide models with smaller generalization
error than series-parallel training under more realistic scenarios concerning noise.
Furthermore, for sequential execution the complexity analysis and
the numerical examples suggest the computational cost is not significantly
different for both methods in typical system identification problems.

Nevertheless, series-parallel training has two real advantages over parallel
training: i)~it seems less likely to be trapped in ``poor'' local
solutions; ii)~it is more amenable to parallelization. 
In~\cite{ribeiro_shooting_2017} a technique called multiple shooting
is introduced in the framework of prediction error methods
as a way of reducing the possibility of parallel training
getting trapped in ``poor'' local minima
and also making the algorithm much more amenable
to parallelization. It seems to be a promising way to solve
the shortcomings of parallel training we have described
in this paper.

\section*{Acknowledgments}
This work has been supported by the Brazilian agencies CAPES, CNPq and FAPEMIG.

\section*{References}
\bibliographystyle{elsarticle-num-names}
\bibliography{freerunsim}

\end{document}